\documentclass[reqno]{amsart}
\title[Counterexamples to the $L^1$ and $L^{\infty}$ boundedness of wave operators] {Counterexamples to the $L^1$ and $L^{\infty}$ boundedness of the one-dimensional wave operators}
\author{Sisi Huang and Xiaohua Yao}

\address{Sisi Huang, Department of Mathematics, Central China Normal University, Wuhan, 430079, P.R. China}
\email{hss@mails.ccnu.edu.cn}

\address{Xiaohua Yao, Department of Mathematics, Key Laboratory of Nonlinear Analysis and Applications(Ministry of Education), and Hubei Key Laboratory of Mathematical Sciences, Central China Normal University, Wuhan, 430079, P.R. China}
\email{yaoxiaohua@ccnu.edu.cn}

\date{\today}
\usepackage{color}
\usepackage{amsmath}
\usepackage{amssymb}
\usepackage[mathscr]{euscript}
\usepackage{mathrsfs}
\usepackage{paralist}
\usepackage{amsthm}
\usepackage{ulem}
\usepackage{bm}
\usepackage{extarrows}
\usepackage{verbatim}
\usepackage{cite}
\usepackage{hyperref}
\usepackage{esint}
\newtheorem{definition}{Definition}[section]
\newtheorem{theorem}[definition]{Theorem}
\newtheorem{lemma}[definition]{Lemma}
\newtheorem{remark}[definition]{Remark}
\newtheorem{proposition}[definition]{Proposition}

\newcommand\R{\mathbb{R}}
\newcommand\Z{\mathbb{Z}}
\newcommand\B{\mathbb{B}}

\newcommand\N{\mathbb{N}}
\newcommand\mcaF{\mathcal{F}}

\newcommand\mcaA{\mathcal{A}}
\newcommand\mcaK{\mathcal{K}}
\newcommand\HH{\mathcal{H}}
\newcommand\mcaR{\mathcal{R}}
\newcommand\mcaS{\mathcal{S}}
\newcommand\mscH{\mathscr{H}}

\newcommand\mcaT{\mathcal{T}}

\numberwithin{equation}{section}

\begin{document}

\maketitle

\begin{abstract} 
It is well established that the wave operators $W_{\pm}(H,-\Delta)$ for the one-dimensional Schr\"{o}dinger operator $H=-\Delta+V(x)$ are bounded on $L^p(\R)$ for all $1<p<\infty$ in both generic and exceptional cases. They are also bounded on $L^1(\R)$ and $L^{\infty}(\R)$ in the exceptional case with $\lim\limits_{x\rightarrow-\infty}f_+(0,x)=1$. For the remaining endpoint cases, it has long been expected that they are generally unbounded at the endpoints $p=1,\infty$ due to the presence of the Hilbert transform in the low energy part, yet a rigorous proof has been missing.

In this paper, we show that even for a bounded and compactly supported non-zero potential $V$, the wave operators $W_{\pm}(H,-\Delta)$ are unbounded on $L^1(\R)$ and $L^{\infty}(\R)$ in the generic case, as well as in the exceptional case with the  condition $\lim\limits_{x\rightarrow-\infty}f_+(0,x)\neq1$. Moreover, in the latter case, they are even unbounded from $L^{\infty}(\R)$ to ${\rm BMO}(\R)$ (Bounded
Mean Oscillation space). Hence together with those known results, our
counterexamples complete the picture of the $L^{p}$ boundedness of
one-dimensional wave operators. 
\end{abstract}

\tableofcontents

\baselineskip=15pt
\section{Introduction and main result}

\subsection{Introduction}
In this paper, we consider the one-dimensional Schr\"{o}dinger operator 
 $$H=-\Delta+V(x),\quad \Delta=\frac{d^2}{dx^2},$$
where $V(x)$ is a real-valued potential satisfying $|V(x)|\lesssim\left<x\right>^{-\beta}$ for some $\beta>2$, with $\left<x\right>=\sqrt{1+|x|^2}$. Under this assumption, by the Kato-Rellich theorem, both $-\Delta$ and $H$ are self-adjoint operators on $L^2(\R)$ with domain $H^{2}(\mathbb R)$ (the Sobolev space of order $2$). They generate the unitary groups $e^{it\Delta}$ and $e^{itH}$, respectively. Moreover, it is well known that $H$ has no positive or zero eigenvalues, all eigenvalues are strictly negative, the absolutely continuous spectrum is $[0,+\infty)$, and the singular continuous spectrum is absent (cf. \cite{Sim15,Kat95}).

The {\it wave operators} associated with $H$ are defined as the following strong limits on $L^2(\R)$:
\begin{equation}\label{definition of wo}
W_{\pm}:=W_{\pm}(H,-\Delta):=s\mbox{-}\lim_{t\to\pm{\infty}}e^{itH}e^{it\Delta}.
\end{equation}
It is known (cf. \cite{Agm75, Kur73, Hor05}) that $W_{\pm}$ exist as partial isometries from $L^2(\R)$ onto $\HH_{ac}(H)$ (the absolutely continuous spectral subspace of $H$), hence are  
$L^2$-bounded. Subsequent works such as Weder \cite{Wed99}, Galtbayar-Yajima \cite{GY00} and D'Ancona-Fanelli \cite{DF06} further established that $W_{\pm}$ are $L^p$ bounded for all $1<p<\infty$ in both generic and exceptional cases (see Definition \ref{defin of generic and excep} below) under various conditions on the potential. Moreover, Weder proved that $W_{\pm}$ are also bounded on $L^1(\R)$ and $L^{\infty}(\R)$ if $V$ is of exceptional type and $\lim\limits_{x\rightarrow-\infty}f_+(0,x)=1$. On the other hand, he pointed out that, in general, they
are not bounded on $L^1(\R)$ or $L^{\infty}(\R)$ because the low energy part of $W_{\pm}$
contains terms with the Hilbert transform $\HH$. 

More precisely, he showed that the low energy part admits the decomposition
\begin{equation}\label{Wed dec}
W_{\pm}\Psi(-\Delta)=W_{\pm,r}\pm \frac{i}{2} \chi_1(x)f_{+}(0,x)\HH \Psi(-\Delta)(a_1+a_2\tau)\mp \frac{i}{2}(1-\chi_1(x))f_{-}(0,x)\HH \Psi(-\Delta)(a_1+a_3\tau),
\end{equation}
where $\Psi\in C^{\infty}_{0}(\R)$ satisfies $\Psi(x)=1$ for $|x|\leq \mu_0$ ($\mu_0>0$ small), $W_{\pm,r}$ are $L^p$ bounded for all $1\le p\le \infty$, $\chi_1(x)\in C^{\infty}(\R)$ with  $\chi_1(x)=0$ for $x\leq0$ and $\chi_1(x)=1$ for $x\geq1$, $f_{\pm}(0,x)$ are Jost functions, $(\tau f)(x):=f(-x)$, and
\begin{equation*}
a_1=T(0)-1,\quad a_2=R_1(0),\quad a_3=R_2(0)
\end{equation*}
with $T(0)$ the transmission coefficient and $R_j(0)$ the reflection coefficients at zero energy. This decomposition offers important insight into the possible endpoint singularities, as it isolates them into terms involving the Hilbert transform, whose coefficients are explicitly expressed through the zero-energy transmission and reflection coefficients. In the exceptional case with $\lim\limits_{x\rightarrow-\infty}f_+(0,x)=1$, all $a_j$ vanish and the decomposition contains no Hilbert transform, which immediately yields the $L^1$ and $L^{\infty}$ boundedness proved by Weder \cite{Wed99}.

However, the Hilbert transform terms persist for the remaining two cases ({\bf the generic case, and the exceptional case with} $\bm{\lim\limits_{x\rightarrow-\infty}f_+(0,x)\neq1}$).
Although the Hilbert transform $\HH$ itself is unbounded on $L^1(\R)$ and $L^{\infty}(\R)$, to the best of our knowledge, a rigorous proof that the composed operators in \eqref{Wed dec}
are unbounded on these two endpoint spaces is not clear. 
First, the spatial multipliers $\chi_{1}(x)f_{\pm}(0,x)$ generally do not commute with
$\mathcal{H}$ and one cannot directly infer the unboundedness
of the composite operators from that of $\mathcal{H}$ alone. Second, the two Hilbert terms are summed, and a possible cancellation between
their singular parts cannot be excluded apriori. 


Motivated by this, we provide a proof for these two remaining cases by a method that avoids the delicate Jost function asymptotics used in the above decomposition. We show that, even for a bounded and compactly supported non-zero potential $V$, $W_{\pm}$ are unbounded on $L^1(\R)$ and $L^{\infty}(\R)$ in the generic case, 
as well as in the exceptional case with 
$\lim\limits_{x\rightarrow-\infty}f_+(0,x)\neq1$. Moreover, in the latter case, 
they are even unbounded from $L^{\infty}(\R)$ to $\mathrm{BMO}(\R)$ (Bounded Mean Oscillation space)  (see Theorem \ref{thm of Main reult}).

To prove this, 
we start from the stationary representation (cf. \cite{Kur73,RS75,RS78}) of $W_{-}$ (the case $W_{+}$ is similar since $W_{+}f=\overline{W_{-}\overline{f}}$), 
$$W_{-}=I-\frac{1}{\pi i}\int_{0}^{+\infty} \lambda R^+_{V}({\lambda}^{2})V(R^+_{0}-R^-_{0})({\lambda}^{2})d\lambda,$$
where $R^{\pm}_{0}(\lambda)$ and $R^{+}_{V}(\lambda)$ are boundary values of free and perturbed resolvents on $(0,+\infty)$, respectively. Fix a sufficiently small constant $0<\lambda_0\ll1$ and choose a smooth cutoff function  $\chi(\lambda)\in C^{\infty}_{0}(\R)$ such that $\chi(\lambda)\equiv1 $ on $(-\lambda_{0},\lambda_{0})$ and  $\chi(\lambda)\equiv0 $ on $(-\infty,-2\lambda_{0})\cup(2\lambda_{0},+\infty)$. We further decompose 
\begin{equation*}
W_{-}=I-\frac{1}{\pi i}(W^L_{-}+W^H_{-}),
\end{equation*}
where the low energy part $W^L_{-}$ and the high energy part $W^H_{-}$ are given by 
\begin{align}
W^{L}_{-}&=\int_{0}^{+\infty}\lambda\chi(\lambda)R^{+}_{V}(\lambda^{2})V(R^+_{0}-R^-_{0})({\lambda}^{2})d\lambda,\notag\\ 
W^{H}_{-}&=\int_{0}^{+\infty}\lambda(1-\chi(\lambda))R^{+}_{V}(\lambda^{2})V(R^+_{0}-R^-_{0})({\lambda}^{2})d\lambda.\label{high energy part}
\end{align}
 It is known that the high energy part of wave operators is $L^p$ bounded for all $1\leq p\leq \infty$ (cf. \cite{Wed99}, another proof the $L^p$ boundedness of $W^H_{-}$ also see Appendix \ref{subsec of high energy part}). Therefore, {\bf in this paper we only focus on the low energy part} $\bm {W^L_{-}}$. To this end, we first study the asymptotic behavior of $R^{+}_V(\lambda^2)$ near $\lambda=0$ (see Subsection \ref{sec of asy of RV}). This enables us to isolate the singular terms at the endpoints $p = 1, \infty$, which exhibit a clearer structure and admit a fairly explicit integral kernel representation (see Subsection \ref{subsec of lp bounds of low energy}). By analyzing these kernels and testing the behaviors of these singular operators on some specific functions, we prove their unboundedness and establish the desired result (see Section \ref{sec of counterex}). 

\subsection{Main result}\label{sec: main result}
To illustrate our results, we introduce some definitions and notation. First recall that the {\it Jost functions} $f_{\pm}(\lambda,x)$ are the solutions to the equation 
 \begin{equation*}
     -f''(\lambda,x)+V(x)f(\lambda,x)=\lambda^2f(\lambda,x),\quad \lambda, x\in\R
 \end{equation*}
 satisfying the following asymptotic conditions
 \begin{equation*}
     |f_{\pm}(\lambda,x)-e^{\pm i\lambda x}|\rightarrow0\ \ {\rm as}\ x\rightarrow\pm \infty.
 \end{equation*}
 It is known (cf. \cite{DT79}) that if $\left<x\right>V(x)\in L^1(\R)$, the Jost functions are uniquely defined for all $\lambda, x\in\R$. We denote by $W(\lambda)$ their Wronskian
 \begin{equation*}
     W(\lambda):=f_+(\lambda,x)f'_-(\lambda,x)-f'_+(\lambda,x)f_-(\lambda,x).
 \end{equation*}
\begin{definition}\label{defin of generic and excep}
We say that $V$ is of generic type if $W(0)\neq0$ and is of exceptional type if $W(0)=0$. We also say that zero is a resonance of $H$ if the potential $V$ is of exceptional type. 
\end{definition}
 Furthermore, Jensen and Nenciu \cite{JN01} pointed out that zero is a resonance of $H$ if and only if there exists a unique (up to a constant factor) non-zero $\phi\in L^{\infty}(\R)$ such that $H\phi=0$ in the distributional sense. Consequently, the trivial potential $V\equiv0$ is of exceptional type.

Moreover, recall that $f\in{\rm BMO}:={\rm BMO}(\R)$ (Bounded Mean Oscillation space) if $f\in L^1_{loc}(\R)$ and $$\|f\|_{{\rm BMO}}:=\sup_{I}\frac{1}{|I|}\int_{I}{}|f(x)-f_{I}|dx<\infty,\quad f_{I}=\frac{1}{|I|}\int_{I}{}f(x)dx,$$
where the supremum takes over all bounded intervals, and $f\in\HH^1:=\HH^1(\R)$ (Hardy space) if $f$ is a tempered distribution and $$\|f\|_{\HH^1}:=\int_{\R}{}\sup_{t>0}|(f*\varphi_{t})(x)|dx<\infty,\quad \varphi_{t}(x)=t^{-1}\varphi(x/t)$$
for some Schwartz\  function $\varphi$ with $\int_{\R}\varphi(x)dx=1$. It is known that the dual space $(\HH^1)^*={\rm BMO}$ (cf.\cite{Fef71}).
    
We write $a\lesssim b$ if $a\leq cb$ for some constant $c>0$ with $a,b\in\R^{+}$. Let $\B(X,Y)$ denote the space of all bounded linear operators from $X$ to $Y$, and abbreviate $\B(X,X)$ as $\B(X)$ when $X=Y$. For brevity, we also abbreviate $L^{p}(\R)$ as $L^p$. Our main results are the following.
\begin{theorem}\label{thm of Main reult}
Let $H=-\Delta+V(x)$ be the Schr\"odinger operator on $\R$ and $V(x)$ be real-valued satisfying $|V(x)|\lesssim\left<x\right>^{-\beta}$ for some $\beta>2$. Then the following statements hold.
\begin{itemize}
\item[{\rm(i)}] If  $V$ is of generic type and $\beta>7$, then $W_{\pm}\notin\B(L^1)\cup\B(L^{\infty})$ but $W_{\pm}\in\B(L^{\infty},{\rm BMO})$. 
 \item[{\rm(ii)}] Assume in addition that $V$ is compactly supported. If $V$ is of exceptional type and $\lim\limits_{x\rightarrow-\infty}f_+(0,x)\neq1$, then $W_{\pm}\notin\B(L^{\infty},{\rm BMO})\cup\B(L^1)$.
\end{itemize}
 \end{theorem}
\begin{remark}\label{rema of main theo}
{\rm Further remarks 
 on Theorem \ref{thm of Main reult} are given as follows.
\begin{itemize}
\item We can explicitly construct a class of potentials satisfying the above conditions (ii). For instance, take $0<\phi\in C^{\infty}(\R)$ satisfying $\phi(x)=b$ for $x>1$ and $\phi(x)=a$ for $x<-1$ ($b\neq a$) and let $V=\frac{\Delta\phi}{\phi}$. It is clear that $ V\in C^{\infty}_0(\R)$, $H\phi=0$, and thus zero is a resonance of $H$. Moreover, in this case the Jost function $f_+(0,x)=\frac{1}{b}\phi(x)$ and satisfies $\lim\limits_{x\rightarrow-\infty}f_+(0,x)\neq1$.
\vskip0.1cm
\item Recall that $W_{\pm}\in\mathcal{B}(L^{p})$ for all $1<p<\infty$ in both generic and exceptional cases (cf. \cite{DF06,GY00,Wed99}), and Weder further proved that $W_{\pm}\in\B(L^1)\cap\B(L^{\infty})$ in the exceptional case with
$\lim\limits_{x\to-\infty}f_{+}(0,x)=1$. Together with these known results, our
counterexamples complete the picture of the $L^{p}$ boundedness of
one-dimensional wave operators. 
\item For both the generic case and the exceptional case with
$\lim\limits_{x\to-\infty}f_{+}(0,x)\neq1$, the wave operators are unbounded on
$L^{1}(\R)$ and $L^{\infty}(\R)$, yet a result of Weder shows that they are bounded
from $\mathcal{H}^{1}(\R)$ to $L^{1}(\R)$. However, the two cases differ in the boundedness from 
$L^{\infty}(\R)$ to $\mathrm{BMO}(\R)$. We prove that $W_{\pm}\in\B(L^{\infty},{\rm BMO})$ in the generic case, while $W_{\pm}\notin\B(L^{\infty},{\rm BMO})$ in the
exceptional case.
\end{itemize}
}
\end{remark}

\begin{remark}\label{background}
{\rm {\bf(The $L^p$ boundedness of wave operators)} We remark that an important property of wave operators is the following {\textit{intertwining property}}
\begin{equation*}
f(H)P_{ac}(H)=W_{\pm}f(-\Delta)W^*_{\pm}
\end{equation*}
which allows one to reduce $L^p-L^q$ estimates for the perturbed operator $f(H)$ to the corresponding estimates for the free operator $f(-\Delta)$:
\begin{equation*}
\|f(H)P_{ac}(H)\|_{L^p\rightarrow L^q}\leq\|W_{\pm}\|_{L^q\rightarrow L^q}\|f(-\Delta)\|_{L^p\rightarrow L^q}\|W^*_{\pm}\|_{L^p\rightarrow L^p},
\end{equation*}
where $f$ is any Borel function on $\R$. Because of such a useful feature, the $L^p$ boundedness of wave operators has been extensively studied. 
It is known that K. Yajima first studied this question in his seminal works \cite{Yaj93,Yaj95a,Yaj95b,Yaj97}
for  Schr\"{o}dinger operators $-\Delta_{\R^d}+V$ on $\R^d$ with $d\geq3$, where he proved that wave operators are $L^p$ bounded for all $1\leq p\leq\infty$ if zero energy is regular. Subsequently, this topic has been developed for lower dimensions, for instance, \cite{Yaj99,JY02,EGG18} for $d=2$ and \cite{Wed99,GY00,DF06,DMW11,Wed22} for $d=1$ and references therein.

Moreover, many studies have extended to higher-order Schr\"{o}dinger operators $(-\Delta_{\R^d})^m+V$ with $m\geq2$. The first work \cite{GG21} by Goldberg and Green established $L^p$ boundedness for $1 < p < \infty$ in the regular case for $(m,d) = (2,3)$, which was subsequently extended to general case $d > 2m$ by Erdo\u{g}an and Green in \cite{EG22,EG23}. Further developments include Mizutani, Wan and Yao's investigation \cite{MWY24a,MWY25} of endpoint behavior and zero energy resonances for $(m,d) = (2,3)$ and their complete analysis \cite{MWY24} of all zero resonance types in $(m,d) = (2,1)$, along with Galtbayar-Yajima's study \cite{GY24} of the $(m,d) = (2,4)$ case. More recently, Erdo\u gan-Green-LaMaster \cite{EGL25} considered the case $d>4m$ while Cheng, Soffer, Wu and Yao \cite{CSWY25,CSWY25b} covered the remaining cases $1 \leq d \leq 4m$.
}    
\end{remark}
\subsection{Organization of the paper}
The remainder of this paper is organized as follows. Section \ref{sec: analysis of Lp bounds of low energy} is devoted to studying the asymptotic behavior of $R^{+}_V(\lambda^2)$ near $\lambda=0$ (Subsection \ref{sec of asy of RV}), and establishing the $L^p$ boundedness of the low energy part $W^L_{-}$, in particular, extracting the singular operators at the endpoints $p=1,\infty$ (Subsection \ref{subsec of lp bounds of low energy}). Based on this analysis, in Section \ref{sec of counterex}, we further show that these singular operators are unbounded by constructing appropriate test functions. Finally, Appendix \ref{proof of asy expansion} provides the detailed proof of asymptotic expansion Lemma \ref{asy expa lemma}, while Appendix 
\ref{subsec of high energy part} presents another proof of the $L^p$ boundedness of the high energy part $W^H_{-}$.

\section{The analysis on the $L^p$ boundedness of the low energy part $W^L_{-}$}\label{sec: analysis of Lp bounds of low energy}
In this section, we analyze the $L^p$ boundedness of the low energy part $W^L_{-}$, which is given by 
\begin{equation}\label{expre of low energy}
W^{L}_{-}=\int_{0}^{+\infty}\lambda\chi(\lambda)R^{+}_{V}(\lambda^{2})V(R^+_{0}-R^-_{0})({\lambda}^{2})d\lambda,
\end{equation}
where $R^{\pm}_{0}(\lambda)$ and $R^{\pm}_{V}(\lambda)$ are boundary values of free and perturbed resolvents on $(0,+\infty)$, respectively, i.e.,
 $$R^{\pm}_{0}(\lambda)=\lim_{\varepsilon\downarrow0}(-\Delta-(\lambda\pm i\varepsilon))^{-1},\quad R^{\pm}_{V}(\lambda)=\lim_{\varepsilon\downarrow0}(H-(\lambda\pm i\varepsilon))^{-1},\quad \lambda>0.$$
 It is known (cf.\cite{Kur73,Agm75}) that such limits exist in the operator norm of $\B(L^2_{s}(\R),L^2_{-s}(\R))$ for $s>\frac{1}{2}$ for both free and perturbed cases under our assumption on $V$. Here 
 $L^2_s(\R):=\{f:\left<x\right>^sf(x)\in L^2(\R)\}.$ Moreover, the kernel of $R^{\pm}_0(\lambda^2)$ is given by (see e.g. \cite{JN01})   
\begin{equation}\label{kernel of free resol}
R^{\pm}_{0}(\lambda^{2},x,y)=\frac{\pm i}{2\lambda}e^{\pm i\lambda|x-y|},\quad \lambda>0,\quad x,y\in\R.
\end{equation}

 Therefore, it is crucial to investigate the asymptotic behavior of $R^+_V(\lambda^2)$ near $\lambda=0$.
\subsection{Asymptotic behavior of $R^+_V(\lambda^2)$ near $\lambda=0$}\label{sec of asy of RV}
To this end, we introduce 
$$M(\lambda):=U+vR^{+}_{0}(\lambda^{2})v,$$
where $U$ and $v$ are the multiplication operators by $U(x)$ and $v(x)$, respectively. Here $U(x)={\rm sgn}~V(x)$, namely, $U(x)=1$ if $V(x)\geq0$ and $U(x)=-1$ if $V(x)<0$, and $v(x)=\sqrt {|V(x)|}$. Since $H$ has no positive eigenvalues in the interval $(0,+\infty)$, it follows that $M(\lambda)$ is invertible on $L^2(\R)$ for all $\lambda>0$. We denote its inverse by $M^{-1}(\lambda)$.
Moreover, using the resolvent identity
$$R^+_V(\lambda^2)=R_0^{+}(\lambda^2)-R^{+}_V(\lambda^2)VR^+_0(\lambda^2)$$
together with the relations $vUv=V$ and $U^2=I$ (the identity operator), we obtain
\begin{equation}\label{relat between RV and M}
R^{+}_{V}(\lambda^{2})V=R^{+}_{0}(\lambda^{2})vM^{-1}(\lambda)v.
\end{equation}
This indicates that we can turn to study the asymptotic behavior of $M^{-1}(\lambda)$ near $\lambda=0$. 

Before presenting the expansion of $M^{-1}(\lambda)$, we introduce some notation. Let $P$ be the orthogonal projection onto the space spanned by $v$, i.e.,
$$P=\alpha\left<\cdot,v\right>v,\quad \alpha=\|V\|^{-1}_{L^1}.$$
Define $Q:=I-P$. Obviously, $Q$ is the orthogonal projection onto $({\rm span}\{v\})^{\bot}$, and thus it satisfies 
$$Qv=0,\quad \left<Qf,v\right>=0,\quad f\in 
 L^2(\R).$$
Denote $$T_0=U+vG_0v,\quad G_0(x,y)=-\frac{1}{2}|x-y|.$$ 
Let $S_1$ be the orthogonal projection onto the kernel space ${\rm Ker}QT_0Q|_{QL^2(\R)}$, that is,
$$S_1L^2(\R)=\{f\in L^2(\R):\left<f,v\right>=0,\ QT_0f=0\}.$$
Throughout this paper, we denote by $K$ the integral operator with the kernel $K(x,y)$, i.e.,
$$(Kf)(x)=\int_{\R}{}K(x,y)f(y)dy.$$
Moreover, we say that an integral operator $K\in\B(L^2)$ is $absolutely\ bounded$ if its associated absolutely value integral operator $|K|$, defined by the kernel $|K(x,y)|$, is also bounded on $L^2(\R)$. 
\begin{lemma}\label{asy expa lemma}
Let $H=-\Delta+V(x)$ with $|V(x)|\lesssim\left<x\right>^{-\beta}$ for some $\beta>2$. Then $M^{-1}(\lambda)$ has the following asymptotic expansion on $L^2(\R)$ for $0<\lambda<2\lambda_0$.
\begin{itemize}
    \item [{\rm (i)}] If $V$ is of generic type and $\beta>7$, then 
    \begin{equation}\label{asy expa reg}
 M^{-1}(\lambda)=QA^0_0Q+\lambda(QA^0_{11}+A^0_{12}Q)+\lambda{P_1}+\Gamma^0_{2}(\lambda).
\end{equation}
\item [{\rm (ii)}] If $V$ is of exceptional type and $\beta>11$, then
\begin{equation}\label{asy expa res}
 M^{-1}(\lambda)=-icS_1\lambda^{-1}+QA^1_{0}Q+2\alpha c(S_1T_0+T_0S_1)+\lambda\left(QA^1_{11}+A^1_{12}Q\right)+\lambda{P_1}+\lambda P_2+\Gamma^1_{2}(\lambda).
\end{equation}
\end{itemize}
Here 
\begin{itemize}
\item  $ {P_1}=-2i\alpha P,\quad P_2=4i\alpha^2c T_0S_1T_0$ with $c$ defined in \eqref{c,c1,c2};
\vskip0.1cm
\item $A^j_0,A^{j}_{11},A^{j}_{12}\ (j=0,1)$ are $\lambda$-independent bounded operators on $L^2(\R)$; 
\vskip0.1cm
\item $\Gamma^{j}_{2}(\lambda)$ are $\lambda$-dependent bounded operators on $L^2(\R)$ such that all the operators appeared in the right hand sides of \eqref{asy expa reg} and \eqref{asy expa res} are absolutely bounded. Moreover, $\Gamma^{j}_{2}(\lambda)$ satisfies that for $s=0,1,2$,
\begin{equation}\label{estim of gamma}
\|(\partial^{s}_{\lambda}\Gamma^{j}_{2})(\lambda)\|_{\B(L^2)}\leq C_s\lambda^{2-s},\quad 0<\lambda<2\lambda_0,\quad j=0,1.
\end{equation}
\end{itemize}
\end{lemma}
\begin{remark}
{\rm
(1) It is known that such an expansion was previously considered by Jensen and Nenciu in \cite{JN01}, but the explicit form of the coefficients was not provided. For our discussion of the $L^p$ boundedness of $W^L_{-}$, a more precise expansion is needed. Based on their ideas, we compute these coefficients in detail in Appendix \ref{proof of asy expansion}.

(2) We remark that the orthogonal projections $Q$ and $S_1$ in the expansions can eliminate the singularity of free resolvent at $\lambda=0$ (see Lemma \ref{cancel lemma} below), which plays a crucial role in establishing the $L^p$ boundedness of $W^L_{-}$.
Furthermore, we point out that if $V$ is of exceptional type, then $S_1$ is a rank-one operator (see Remark \ref{rema of S1 and rela betwe c2 and lim}), which greatly facilitates our subsequent discussion of boundedness at the endpoints $p=1,\infty$. 
}
\end{remark}
\subsection{The $L^p$ boundedness of $W^L_{-}$}\label{subsec of lp bounds of low energy} With Lemma \ref{asy expa lemma}, we can derive the following $L^p$ boundedness of $W^L_{-}$.
\begin{theorem}\label{theorem of lp bounds of low energy} 
Let $H=-\Delta+V$ with $|V(x)|\lesssim \left<x\right>^{-\beta}$ for some $\beta>2$. Then the following statements hold.    \begin{itemize}
    \item [{\rm(i)}] If $V$ is generic type and $\beta>7$, then 
\begin{equation}\label{deco of regular}
W^{L}_{-}=\sum\limits_{A\in\Omega_0 }K_{A}+K_{{P_1}}
\end{equation}
and $\{K_A:A\in\Omega_0\}\subseteq\B(L^p)$ for all $1\leq p\leq\infty$, while 
$$ K_{{P_1}}\in \B(L^p)\cap\B(\HH^1,L^1)\cap\B(L^{\infty},{\rm BMO})\ for\ all\ 1<p<\infty.$$
Therefore, $$W^L_{-}\in\B(L^p)\cap\B(\HH^1,L^1)\cap\B(L^{\infty},{\rm BMO})\ for\ all\ 1<p<\infty.$$ 
\item [{\rm(ii)}] If $V$ is exceptional type and $\beta>11$, then 
\begin{equation}\label{deco of resonance}
W^{L}_{-}=\sum\limits_{A\in\Omega_1}K_A+K_{-1}+\sum\limits_{j=1}^{2}(K_{0j}+K_{{P_{j}}})   
\end{equation}
and $\{K_A:A\in\Omega_1\}\subseteq\B(L^p)$ for all $1\leq p\leq\infty$, while  
\begin{align*}
K_{-1}&\in\B(L^p)\cap\B(\HH^1,L^1)\ for\ all\ 1<p<\infty,\\
\{K_{01},K_{02},K_{{P_1}},K_{{P_2}}\}&\subseteq \B(L^p)\cap\B(\HH^1,L^1)\cap\B(L^{\infty},{\rm BMO})\ for\ all\ 1<p<\infty.
\end{align*}
Therefore,  $$W^L_{-}\in\B(L^p)\cap\B(\HH^1,L^1)\ for\ all\ 1<p<\infty.$$ 
Here $\Omega_j=\{QA^j_0Q,\lambda QA^j_{11},\lambda A^j_{12}Q, \Gamma^j_2(\lambda)\}$ and
\begin{align}\label{kernels in low energy}
\begin{split}
K_A(x,y)&=\int_{0}^{+\infty}\lambda\chi(\lambda)\big[R^+_0(\lambda^{2})vAv(R^+_0-R^-_0)(\lambda^{2})\big](x,y)d\lambda,\quad A\in\Omega_j,\quad j=0,1,\\
K_{-1}(x,y)&=-ic\int_{0}^{+\infty}\chi(\lambda)\big[R^+_0(\lambda^{2})vS_1v(R^+_0-R^-_0)(\lambda^{2})\big](x,y)d\lambda,\\
K_{01}(x,y)&=2\alpha c\int_{0}^{+\infty}\lambda\chi(\lambda)\big[R^+_0(\lambda^{2})vS_1T_0v(R^+_0-R^-_0)(\lambda^{2})\big](x,y)d\lambda,\\
K_{02}(x,y)&=2\alpha c\int_{0}^{+\infty}\lambda\chi(\lambda)\big[R^+_0(\lambda^{2})vT_0S_1v(R^+_0-R^-_0)(\lambda^{2})\big](x,y)d\lambda,\\
K_{P_{j}}(x,y)&=\int_{0}^{+\infty}\lambda^2\chi(\lambda)\big[R^+_0(\lambda^{2})vP_jv(R^+_0-R^-_0)(\lambda^{2})\big](x,y)d\lambda,\quad j=1,2.
\end{split}
\end{align}
\end{itemize}
\end{theorem}
\begin{remark}\label{rem of bounds of Wpm}
{\rm As a consequence of this theorem and Theorem \ref{bounds of high energy}, we obtain 
\begin{itemize}
\item If $V$ is generic type and $\beta>7$, then $W_{\pm}\in\B(L^p)\cap\B(\HH^1,L^1)\cap\B(L^{\infty},{\rm BMO})$ for all $1<p<\infty$.
\vskip0.1cm
\item If $V$ is exceptional type and $\beta>11$, then $W_{\pm}\in\B(L^p)\cap\B(\HH^1,L^1)$ for all $1<p<\infty$.
\end{itemize} 
}
\end{remark}
To derive Theorem \ref{theorem of lp bounds of low energy}, we first combine \eqref{expre of low energy} and \eqref{relat between RV and M} to rewrite  \begin{equation*}
W^{L}_{-}=\int_{0}^{+\infty}\lambda\chi(\lambda)R^{+}_{0}(\lambda^{2})vM^{-1}(\lambda)v(R^+_{0}-R^-_{0})({\lambda}^{2})d\lambda.
\end{equation*} 
Substituting 
\eqref{asy expa reg} and \eqref{asy expa res} into this expression immediately yields \eqref{deco of regular} and \eqref{deco of resonance}, respectively. Hence, the key step is to establish the boundedness stated above. For this purpose, we present the following cancellation lemma.
\begin{lemma}\label{cancel lemma}
Let $\lambda>0$. Then for any $T=Q,S_1$, we have
\begin{equation*}
 [R^{\pm}_0(\lambda^2)vTf](x)=\int_{\R}\mcaR^{\pm}(\lambda,x,y)(Tf)(y)dy,\quad\ [TvR^{\pm}_0(\lambda^2)f](x)=T(\mcaF^{\pm}(\lambda,x)),
\end{equation*}
where 
\begin{equation*}
\mcaR^{\pm}(\lambda,x,y)=\frac{1}{2}yv(y)\int_{0}^{1}{\rm sgn}(x-\theta y)e^{\pm i \lambda|x-\theta y|}d\theta,\quad \mcaF^{\pm}(\lambda,x)=\int_{\R}\mcaR^{\pm}(\lambda,y,x)f(y)dy.
\end{equation*}
\end{lemma}
 This lemma shows that, compared to $R^{\pm}_0(\lambda^2)=O(\lambda^{-1})$~(here $O(\lambda^{-1})$ refers to the order of the kernel $R^{\pm}_{0}(\lambda^2,x,y)$ with respect to $\lambda$ and the same convention applies to the operators below unless stated otherwise), the operators considered in this lemma can decrease the singularity near $\lambda=0$. Precisely, 
$$ R^{\pm}_{0}(\lambda^{2})vT=O(1),\quad TvR^{\pm}_{0}(\lambda^{2})=O(1),\quad  T=Q,S_1,\quad\lambda\rightarrow 0^{+}.$$
\begin{proof}[Proof of Lemma {\rm\ref{cancel lemma}}] Firstly, for $T=Q,S_1$, it follows from \eqref{kernel of free resol} that 
\begin{equation}\label{TvR0}
[R^{\pm}_{0}(\lambda^{2})vTf](x)=\frac{\pm i}{2\lambda}\int_{\R}^{}e^{\pm i \lambda|x-y|}v(y)(Tf)(y)dy,\ [TvR^{\pm}_{0}(\lambda^{2})f](x)=\frac{\pm i}{2\lambda}T\big(v(x)\int_{\R}^{}e^{\pm i \lambda|x-y|}f(y)dy\big) 
.
\end{equation}
Applying the Taylor expansion to $e^{\pm i\lambda|x-y|}$ yields
\begin{equation}\label{taylor expan}
e^{\pm i\lambda|x-y|}=e^{\pm i\lambda|x|}\mp i\lambda y\int_{0}^{1}{\rm sgn}(x-\theta y)e^{\pm i\lambda|x-\theta y|}d\theta=e^{\pm i\lambda|y|}\mp i\lambda x\int_{0}^{1}{\rm sgn}(y-\theta x)e^{\pm i\lambda|y-\theta x|}d\theta.
\end{equation}
Substituting the first (resp. second) expression of \eqref{taylor expan} into the first (resp. second) formula of \eqref{TvR0} and using the orthogonal condition $\left<Tf,v\right>=0$ (resp. $Tv=0$) for $T=Q,S_1$, we obtain the desired result.
\end{proof}
Using this lemma, we can classify the operators in \eqref{deco of regular} and \eqref{deco of resonance} into the following two groups according to the order of their kernels with respect to $\lambda$ as $\lambda\rightarrow0^{+}$:
\begin{equation*}
O(\lambda):K_A~(A\in\Omega_0\cup\Omega_1),\quad O(1):K~(K=K_{-1},K_{01},K_{02}, K_{P_1}, K_{P_2}).
\end{equation*}

The $L^p$ boundedness of $W^L_{-}$ consequently reduces to proving the boundedness of these two operator classes. In what follows, we shall prove Theorem \ref{theorem of lp bounds of low energy} through two Propositions \ref{good1} and \ref{bounds of Ptuta}. 

First, we deal with the operators in the class $O(\lambda)$. Prior to this, we state the following Schur test lemma, which will be used frequently to establish the $L^p$ boundedness of integral operators.
\begin{lemma}\label{shur test}
If the kernel $K(x,y)$ satisfies $$\sup_{x\in\R}\int|K(x,y)|dy+\sup_{y\in\R}\int|K(x,y)|dx<\infty,$$ then $K\in \B(L^{p}(\R))$ for all $1\leq p\leq \infty.$
\end{lemma}
In particular, the estimate $|K(x,y)|\lesssim \left<|x|-|y|\right>^{-\gamma}$ with $\gamma>1$ satisfies the condition of this lemma and will be used frequently in the proof.

\begin{proposition}\label{good1}
Under the assumption of Theorem {\rm\ref{theorem of lp bounds of low energy}}, let $\Omega_j$ and $K_A$ be as in \eqref{kernels in low energy}. Then for any $j=0,1$ and $A\in\Omega_j$, we have $K_A\in\B(L^p)$ for all $\ 1\leq p\leq\infty$.
\end{proposition}
\begin{proof} Recall that $\Omega_j=\{QA^j_0Q,\lambda QA^j_{11},\lambda A^j_{12}Q,  \Gamma^j_2(\lambda)\}$, it suffices to consider the case $A\in\Omega_0$. Moreover, we only handle $K_A$ for $A=QA^0_0Q$ and $A=\Gamma^0_2(\lambda)$, since the remained cases are similar.
\vskip0.2cm
\underline{(1) For $A=QA^0_0Q$}, denote
\begin{equation*}
K_0(x,y)=\int_{0}^{+\infty}\lambda\chi(\lambda)\big[R^+_0(\lambda^{2})vQA^0_0Qv(R^+_0-R^-_0)(\lambda^{2})\big](x,y)d\lambda.
\end{equation*}
By Lemma \ref{cancel lemma}, we have 
\begin{equation*}
    K_0(x,y)=\frac{1}{4}\sum\limits_{\pm}K^{\pm}_0(x,y),
\end{equation*}
where $\Theta=(u_1,u_2,\theta_1,\theta_2),\ X_1=x-\theta_1u_1,\ Y_2=y-\theta_2u_2$ and 
\begin{align*}
K^{\pm}_0(x,y)&=\pm\int_{0}^{+\infty}e^{i\lambda(|x|\pm|y|)}\lambda\chi(\lambda)k^{\pm}_0(\lambda,x,y)d\lambda,\\
k^{\pm}_0(\lambda,x,y)&=\int_{\R^2\times[0,1]^2}({\rm sgn}X_1)({\rm sgn}Y_2)e^{i\lambda(|X_1|-|x|\pm (|Y_2|-|y|))}u_1u_2(vQA^0_0Qv)(u_1,u_2)d\Theta.
\end{align*}
We shall prove \begin{equation}\label{esti of Kpm0}
    |K^{\pm}_{0}(x,y)|\lesssim\left<|x|-|y|\right>^{-2},\quad \forall\  x,y\in \R.
\end{equation}
Together with Lemma \ref{shur test} this yields $K_0\in \B(L^p)$ for all $1\leq p\leq\infty$. First, we note that $K^{\pm}_0(x,y)$ is uniformly bounded since $\chi(\lambda)\in C^{\infty}_0(\R)$ and 
\begin{equation*}
|k^{\pm}_0(\lambda,x,y)|\lesssim \int_{\R}\left<u_1\right>|v(u_1)|\cdot |QA^0_0Q|(\left<\cdot\right>|v(\cdot)|)(u_1)du_1 \lesssim\|\left<u_1\right>^2V(u_1)\|_{L^1} <\infty. 
\end{equation*}
Thus, if $||x|-|y||\leq1$, then $|K^{\pm}_{0}(x,y)|\lesssim1\lesssim\left<|x|-|y|\right>^{-2}.$ Now suppose $||x|-|y||\geq1$ and set $$\Phi^{\pm}(x,y,\Theta)=i(|X_1|-|x|\pm (|Y_2|-|y|)).$$
By the triangle inequality, $\left|\Phi^{\pm}(x,y,\Theta)\right|\leq |u_1|+|u_2|$. Hence, for any $\ell=0,1,2$, 
\begin{equation}\label{estimate of k_0}
|(\partial^{\ell}_{\lambda}k^{\pm}_0)(\lambda,x,y)|\lesssim
\int_{\R^2}^{}\left<u_1\right>^{\ell+1}|(vQA^0_0Qv)(u_1,u_2)|\left<u_2\right>^{\ell+1}du_1du_2\lesssim \|\left<u_1\right>^{2\ell+2}V(u_1)\|_{L^1} <\infty.
\end{equation} 
Let $L^{\pm}(\lambda,x,y)=\lambda \chi(\lambda)k^{\pm}_0(\lambda,x,y)$. Integrating by parts twice gives
\begin{align}\label{inte parts twice of Kpm0}
\begin{split}
K^{\pm}_{0}(x,y)&=\pm\left[\frac{e^{i\lambda(|x|\pm |y|)}}{i(|x|\pm |y|)}L^{\pm}(\lambda,x,y)\Big|^{+\infty}_{\lambda=0}-\int_{0}^{+\infty}\frac{e^{i\lambda(|x|\pm |y|)}}{i(|x|\pm |y|)}(\partial_{\lambda}L^{\pm})(\lambda,x,y)\right]\\
&=\frac{\pm}{(|x|\pm |y|)^2}\left(e^{i\lambda(|x|\pm |y|)}(\partial_{\lambda}L^{\pm})(\lambda,x,y)\Big|^{+\infty}_{\lambda=0}-\int_{0}^{+\infty}e^{i\lambda(|x|\pm |y|)}(\partial^2_{\lambda}L^{\pm})(\lambda,x,y
)d\lambda\right)\\
&=O(\left<|x|-|y|\right>^{-2}),
\end{split}
\end{align}
where the second equality uses $L^{\pm}(\lambda,x,y)\big|_{\lambda=0,+\infty}=0$, and the last step follows from $\chi(\lambda)\in C^{\infty}_0(\R)$ and \eqref{estimate of k_0}.
This establishes \eqref{esti of Kpm0}.
\vskip0.3cm
\underline{(2) For $A=\Gamma^0_2(\lambda)$}, let 
\begin{equation*}
K_r(x,y)=\int_{0}^{+\infty}\lambda\chi(\lambda)\big[R^+_0(\lambda^{2})v\Gamma^0_2(\lambda)v(R^+_0-R^-_0)(\lambda^{2})\big](x,y)d\lambda.
\end{equation*}
We can write
\begin{equation*}
    K_r(x,y)=-\frac{1}{4}\sum\limits_{\pm}K^{\pm}_r(x,y),
\end{equation*}
where $X_1=x-u_1,\ Y_2=y-u_2,\ {\Gamma}(\lambda)=\lambda^{-2}\Gamma^0_2(\lambda)$ and 
\begin{align*}
K^{\pm}_r(x,y)&=\int_{0}^{+\infty}e^{i\lambda(|x|\pm|y|)}\lambda\chi(\lambda)k^{\pm}_r(\lambda,x,y)d\lambda,\\
k^{\pm}_r(\lambda,x,y)&=\int_{\R^2}e^{i\lambda(|X_1|-|x|\pm (|Y_2|-|y|))}(v\Gamma(\lambda)v)(u_1,u_2)du_1du_2.
\end{align*}
\indent Compared with $K_0$, the additional complexity comes from $k^{\pm}_r$. To handle it we introduce a homogeneous dyadic partition of unity $\{{\varphi_{N}}\}_{N\in \Z}$ on $(0,+\infty)$: choose $\varphi\in C^{\infty}_{0}(\R^{+})$ with $0\leq\varphi\leq1$, supp$~\varphi\subset[\frac{1}{4},1],$ set $\ \varphi_{N}(\lambda)=\varphi(2^{-N}\lambda)$ so that supp$~\varphi_{N}\subset[2^{N-2},2^{N}]$, and $$ \sum_{N\in\Z}{}\varphi_{N}(\lambda)=1,  \quad\lambda>0.$$
Let $N_0=[\log_22\lambda_0]$. Combining the supports of $\chi(\lambda)$ and $\varphi_N(\lambda)$, we decompose
\begin{equation*}
\chi(\lambda)=\sum_{N=-\infty}^{N_0}{\chi}_{N}(\lambda),\quad {\chi}_{N}(\lambda)=\chi(\lambda)\varphi_{N}(\lambda),\quad\lambda>0.
\end{equation*}
Consequently,
\begin{equation*}
K^{\pm}_{r}(x,y)=\sum_{N=-\infty}^{N_0}\int_{0}^{+\infty}e^{i\lambda(|x|\pm |y|)}\lambda{\chi}_{N}(\lambda)k^{\pm}_r(\lambda,x,y)d\lambda:=\sum_{N=-\infty}^{N_0}K^{\pm}_{r,N}(x,y).
\end{equation*}
We claim that for every  $N\leq N_0$,
\begin{equation}\label{estim of krlN}
|K^{\pm}_{r,N}(x,y)|\lesssim \min\{{2^{2N},\left<|x|-|y|\right>^{-2}}\},\quad \forall\ x,y\in\R.
\end{equation}
Once this is proved, it implies for any $\theta\in[0,1]$, 
 $$|K^{\pm}_{r,N}(x,y)|\lesssim2^{2N(1-\theta)}\left<|x|-|y|\right>^{-2\theta},\quad \forall\ x,y\in\R.$$
Choosing $\theta=\frac{3}{4}$ gives
$$|K^{\pm}_{r}(x,y)|\lesssim\sum_{N=-\infty}^{N_0}2^{\frac{N}{2}}\left<|x|-|y|\right>^{-\frac{3}{2}}\lesssim\left<|x|-|y|\right>^{-\frac{3}{2}},\quad \forall\ x,y\in\R,$$
and Lemma \ref{shur test} yields $K_r\in\B(L^p)$ for all $1\leq p\leq\infty$. To derive \eqref{estim of krlN}, we first recall from \eqref{estim of gamma} that
$$\|\Gamma^{(j)}(\lambda)\|_{\B(L^2)}=O(\lambda^{-j}),\quad 0<\lambda\leq2\lambda_0.$$
This gives, for any $\ell=0,1,2$,
\begin{equation}\label{estim of krN}
\sup_{x,y\in\R}\left|(\partial^{\ell}_{\lambda}k^{\pm}_r)(\lambda,x,y)\right|\lesssim\lambda^{-\ell}\big\|\left<\cdot\right>^{2\ell}V(\cdot)\big\|_{L^1(\R)},\quad  0<\lambda\leq2\lambda_0.
\end{equation}
Now fix $ N\leq N_0$. Because supp$\ \chi_N\subseteq [0,2\lambda_0]\cap[2^{N-2},2^N]$, 
\begin{align}\label{klbn}
|K^{\pm}_{r,N}(x,y)|=\left|\int_{0}^{+\infty}e^{i\lambda(|x|\pm |y|)}\lambda{\chi}_{N}(\lambda)k^{\pm}_r(\lambda,x,y)d\lambda\right|\lesssim\int_{2^{N-2}}^{2^{N}}\lambda d\lambda\lesssim2^{2N},\quad \forall \ x,y\in\R.
\end{align}
Hence, if $||x|-|y||\leq1$, then
$$|K^{\pm}_{r,N}(x,y)|\lesssim2^{2N}\lesssim2^{2N_0} \left<|x|-|y|\right>^{-2}\lesssim \left<|x|-|y|\right>^{-2},\quad \forall\ N\leq N_0 ,$$
If $||x|-|y||\geq1$, set $L^{\pm}_{N}(\lambda,x,y)=\lambda{\chi}_{N}(\lambda)k^{\pm}_r(\lambda,x,y)$.
Apply integration by parts twice to $K^{\pm}_{r,N}(x,y)$ and using ${\chi}^{(\ell)}_N(\lambda)\big|_{\lambda=0,+\infty}=0$ together with \eqref{estim of krN}, we have
\begin{align*}
\big|K^{\pm}_{r,N}(x,y)\big|=\frac{1}{(|x|\pm|y|)^2}\left|\int_{0}^{+\infty}e^{i\lambda(|x|\pm |y|)}(\partial^2_{\lambda}L^{\pm}_{N})(\lambda,x,y)d\lambda\right|\lesssim\frac{1}{(|x|\pm |y|)^2}\int_{2^{N-2}}^{2^N}\big|(\partial^2_{\lambda}L^{\pm}_{N})(\lambda,x,y)\big|d\lambda.
\end{align*}
Since
$\big|{\chi}^{(s)}_{N}(\lambda)\big|\lesssim2^{-Ns}$ for any $s\in\N$,
this bound together with \eqref{estim of krN} implies 
$$\big|(\partial^2_{\lambda}L^{\pm}_{N})(\lambda,x,y)\big|\lesssim 2^{-N},\quad \lambda\in[2^{N-2},2^N]$$
and therefore $\big|K^{\pm}_{r,N}(x,y)\big|\lesssim \left<|x|-|y|\right>^{-2}.$ The claim \eqref{estim of krlN} follows.
\end{proof}
Next, we turn to establish the boundedness of operators in the class $O(1)$. For these integral operators, we shall need the following key lemma.
\begin{lemma}\label{C-Z lemma}
{\rm(\cite[Lemma 3.3 and Lemma 3.6]{MWY24})}
{\rm (i)} Let $\phi \in C^{\infty}(\R,\R)$ be such that $\phi(s)=0$ for $0\leq s\leq1$ and $\phi(s)=1$ for $s\geq2$. Define functions $k^{\pm}(x,y)$ by
\begin{equation*}
k^{\pm}(x,y)=\frac{\phi(||x|\pm|y||^2)}{|x|\pm |y|},
\end{equation*}
 then  ${k^{\pm}}\in\B(L^p)$ for all $1<p<\infty$. 
\vskip0.2cm
{\rm(ii)} Let $k^{\pm}$ be as in {\rm(i)} and let $g^{\pm},g_1,g_2$ be integral operators with the following kernels:
\begin{align*}
g^{\pm}(x,y)&:=k^{+}(x,y)\pm k^{-}(x,y);\\
g_{1}(x,y)&:=({\rm sgn}~x)g^+(x,y);\\
g_{2}(x,y)&:=({\rm sgn}~x)({\rm sgn}~y)g^{-}(x,y),
\end{align*}
then $g^{\pm},g_1\in\B(\HH^1,L^1)\cap\B(L^{\infty},{\rm BMO})$ and $g_2\in \B(\HH^1,L^1)$.
\end{lemma}
\begin{proposition}\label{bounds of Ptuta}
Under the assumption of Theorem {\rm\ref{theorem of lp bounds of low energy}}, let $K_{-1},K_{0j},K_{P_j}$ be as in \eqref{kernels in low energy}. Then we have
\begin{itemize}
    \item [{\rm (1)}] for every $K\in\{K_{-1},K_{01},K_{02},K_{{P_1}},K_{{P_2}}\}$, $K\in\B(L^p)$ for all $1<p<\infty$.
    \item [{\rm (2)}] $K_{-1}\in\B(\HH^1,L^1)$, and for any $K\in \{K_{01},K_{02},K_{{P_1}},K_{{P_2}}\}$, $K\in \B(\HH^1,L^1)\cap\B(L^{\infty},{\rm BMO})$.
\end{itemize}
\end{proposition}
\begin{proof}
\underline{Step 1: For $K_{P_j}$ with $j=1,2$.} We have 
\begin{equation*}
K_{{P_j}}(x,y)
=-\frac{1}{4}\sum\limits_{\pm}
\int_{0}^{+\infty}e^{i\lambda(|x|\pm|y|)}\chi(\lambda)k^{\pm}_{{P_j}}(\lambda,x,y)d\lambda:=-\frac{1}{4}\sum\limits_{\pm}K^{\pm}_{{P_j}}(x,y),
\end{equation*}
where $X_1=x-u_1,\ Y_2=y-u_2$ and
\begin{equation*}
k^{\pm}_{{P_j}}(\lambda,x,y)=\int_{\R^2}e^{i\lambda(|X_1|-|x|\pm(|Y_2|-|y|))}(v {P_j}v)(u_1,u_2)du_1du_2.
\end{equation*}
It is clear that $|K^{\pm}_{{P_j}}(x,y)|\lesssim 1\lesssim \left<|x|-|y|\right>^{-2}$ if $||x|-|y||\leq1$, because
\begin{equation*}
\sup_{\lambda,x,y}\big|(\partial^{\ell}_{\lambda}k^{\pm}_{{P_j}})(\lambda,x,y)\big|\lesssim\big\|\left<\cdot\right>^{2\ell}V(\cdot)\big\|_{L^1(\R)},\quad \ell=0,1,2.
\end{equation*}
For $||x|-|y||\geq1$, we decompose $K^{\pm}_{P_j}(x,y)$ as 
\begin{equation*}
K^{\pm}_{P_j}(x,y)=\phi_{\pm}K^{\pm}_{P_j}(x,y)+(1-\phi_{\pm})K^{\pm}_{P_j}(x,y),
\end{equation*}
where $\phi_{\pm}:=\phi\left(||x|\pm|y||^2\right)$ with $\phi$ as in Lemma \ref{C-Z lemma}. 
For the first term, by an  argument similar to the one for $K^{\pm}_0(x,y)$ in \eqref{inte parts twice of Kpm0}, we obtain
\begin{equation*}
K^{\pm}_{P_j}(x,y)=ik^{\pm}(x,y)\left<v,{P_j}v\right>+O\big(\left<|x|-|y|\right>^{-2}\big).
\end{equation*}
For the second term, the definition of $\phi$ and uniform boundedness of $K^{\pm}_{P_j}(x,y)$ gives
\begin{equation*}
(1-\phi_{\pm})K^{\pm}_{P_j}(x,y)=O(\left<|x|-|y|\right>^{-2}).
\end{equation*}
Hence, using $P=\alpha\left<\cdot,v\right>v$ and $S_1=\left<\cdot,\omega\right>\omega$ with $\|\omega\|_{L^2}=1$, we derive
\begin{equation}\label{kernel of KPj}
K_{P_j}(x,y)=\left\{\begin{aligned}
&-\frac{1}{2}g^+(x,y)+O(\left<|x|-|y|\right>^{-2}),\quad j=1,\\
&|c_1|^2cg^+(x,y)+O(\left<|x|-|y|\right>^{-2}), \quad j=2.
\end{aligned}\right.
\end{equation}
Combining Lemma \ref{shur test} and Lemma \ref{C-Z lemma} yields $K_{P_j}\in\B(L^p)\cap\B(\HH^1,L^1)\cap\B(L^{\infty},{\rm BMO})$.
\vskip0.3cm
\underline{Step 2: For $K\in\{K_{-1},K_{01},K_{02}\}$.} Based on Lemma \ref{cancel lemma} and the method used for $K_{P_j}$, we obtain
\begin{align}\label{kernels of singular terms}
\begin{split}
    K_{-1}(x,y)&=\frac{c}{4}a(x)\overline{a(y)}g^{-}(x,y)+O(\left<|x|-|y|\right>^{-2}),\\
    K_{01}(x,y)&=-\frac{\overline{c_1}c}{2}a(x)g^{+}(x,y)+O(\left<|x|-|y|\right>^{-2}),\\
 K_{02}(x,y)&=-\frac{{c_1}c}{2}\overline{a(y)}g^{-}(x,y)+O(\left<|x|-|y|\right>^{-2}), 
\end{split}
\end{align}
where 
\begin{equation}\label{ax}
    a(x)=\int_{\R}\int_{0}^{1}({\rm sgn}(x-\theta y))d\theta yv(y)\omega(y)dy.
\end{equation}
Since $a(x)\in L^{\infty}(\R)$, it follows from Lemma \ref{C-Z lemma} and  Minkowski's inequality that $K\in \B(L^p)$ for every $K\in\{K_{-1},K_{01},K_{02}\}$ and $1<p<\infty$. It remains to prove 
$$K_{-1}\in\B(\HH^1,L^1)\ {\rm and}\  K_{01},K_{02}\in\B(\HH^1,L^1)\cap\B(L^{\infty},{\rm BMO}).$$ For $K_{-1}$, using Lemma \ref{cancel lemma} and Fubini's theorem, we have
\begin{equation}\label{kernel of Kminus1}
K_{-1}(x,y)=\frac{-ic}{4}\int_{\R^2\times[0,1]^2}^{}u_1u_2(vS_1v)(u_1,u_2)k_{-1}(X_1,Y_2)d\Theta,
\end{equation}
where $\Theta=(u_1,u_2,\theta_1,\theta_2),\ X_1=x-\theta_1u_1,\ Y_2=y-\theta_2u_2$ and 
\begin{align*}
k_{-1}(x,y)&=\sum\limits_{\pm}\pm({\rm sgn}~x)({\rm sgn}~y)\int_{0}^{+\infty}\chi(\lambda)e^{i\lambda(|x|\pm|y|)}d\lambda.
\end{align*}
By the same method as for $K_{P_j}$, we further obtain  
$$k_{-1}(x,y)=ig_2(x,y)+O(\left<|x|-|y|\right>^{-2}).$$
This, together with Lemma \ref{C-Z lemma} gives  $k_{-1}\in\B(\HH^1,L^1)$.
Since the $L^1$ norm and $\HH^{1}$ norm are invariant under translation, we deduce $K_{-1}\in\B(\HH^1,L^1)$. Indeed, for any $f\in\HH^1(\R)$, let $(\mcaT_b f)(x):=f(x-b)$, then
\begin{align*}
\|K_{-1}f\|_{L^1}&=\frac{c}{4}\int_{\R}\left|\int_{\R^2\times[0,1]^2}u_1u_2(vS_1v)(u_1,u_2)(k_{-1}\mcaT_{-\theta_2u_2} f)(x-\theta_1u_1)d\Theta\right|dx\\
&\lesssim \int_{\R^2\times[0,1]^2}|u_1u_2(vS_1v)(u_1,u_2)|\cdot\|k_{-1}\mcaT_{-\theta_2u_2} f\|_{L^1}d\Theta\\
&\lesssim \int_{\R^2\times[0,1]^2}|u_1u_2(vS_1v)(u_1,u_2)|\cdot\|\mcaT_{-\theta_2u_2} f\|_{\HH^1}d\Theta\lesssim \|f\|_{\HH^1}
\end{align*}
A similar argument for  $K_{-1}$ yields
\begin{align*}
K_{01}(x,y)&=-\frac{1}{4}\int_{\R^2\times[0,1]}u_1(vS_1T_0v)(u_1,u_2)k_{01}(X_1,\widetilde{Y}_2)d{\Theta}_1,\quad {\Theta}_1=(u_1,u_2,\theta_1),\\
K_{02}(x,y)&=-\frac{1}{4}\int_{\R^2\times[0,1]}u_2(vT_0S_1v)(u_1,u_2)k_{02}(\widetilde{X}_1,{Y}_2)d{\Theta}_2,\quad {\Theta}_2=(u_1,u_2,\theta_2)
\end{align*}
where $X_1=x-\theta_1u_1$, $\widetilde{X}_1=x-u_1$, ${Y}_2=y-\theta_2u_2$, $\widetilde{Y}_2=y-u_2$, and
$$k_{01}(x,y)=g_1(x,y)+O(\left<|x|-|y|\right>^{-2}),\quad k_{02}(x,y)=\overline{g_1(y,x)}+O(\left<|x|-|y|\right>^{-2}).$$
Therefore, by Lemma \ref{C-Z lemma}, we conclude $K_{01},K_{02}\in \B(\HH^1,L^1)\cap\B(L^{\infty},{\rm BMO})$.
\end{proof}
To sum up, combining Propositions \ref{good1} and \ref{bounds of Ptuta}, we complete the proof of Theorem \ref{theorem of lp bounds of low energy}. 
\section{Counterexamples to the $L^{1}$ and $L^{\infty}$ boundedness of $W^L_{-}$} \label{sec of counterex}
As shown in Theorem \ref{theorem of lp bounds of low energy}, the unboundedness of $W^L_{-}$ on $L^1(\R)$ and $L^{\infty}(\R)$ reduces to the following operators $\mcaK$. 
\begin{itemize}
\item If $V$ is of generic type , $\mcaK=K_{P_1}\in\B(L^p)$ for $1<p<\infty$.
\item If $V$ is of exceptional type, $\mcaK=K_{-1}+K_{01}+K_{02}+K_{P_1}+K_{P_2}\in\B(L^p)$ for $1<p<\infty$ with
$$K_{0j},K_{P_j}\in\B(L^{\infty},{\rm BMO}),\quad j=1,2.$$
\end{itemize}
For this purpose, let  $\chi_{\Omega}(x)$ be the characteristic function on the set $\Omega$. We will choose the test functions
$$f_{r}(x):=\chi_{[-r,r]}(x),\quad h_r(y):=({\rm sgn}\ y)\chi_{\{r\leq|y|\leq2r\}}(y)\in \HH^1(\R)$$
to disprove the boundedness at the endpoints $p=1,\infty$. Precisely, we have
\begin{proposition}
Let $H=-\Delta+V$ with $|V(x)|\lesssim \left<x\right>^{-\beta}$ for $\beta>2$. Then the following statements hold.
\begin{itemize} 
\item [{\rm(1)}] If $V$ is of generic type and $\beta>7$, then
$${K_{P_1}}f_{r}\notin L^1(\R)\ for\ any\ r>0\ and\  |(K_{P_1}f_{r})(r+2)|\rightarrow+\infty\ as\ \ r\rightarrow +\infty.$$
Thus, $W^L_{-}\not\in\B(L^1)\cup\B(L^{\infty})$.
\item  [{\rm(2)}] Assume in addition that $supp \ V\subseteq\{x\in\R: |x|\leq N-1\}$ for some integer $N\geq2$. If V is of exceptional type and $\lim\limits_{x\rightarrow-\infty}f_{+}(0,x)\neq1$, then 
\begin{equation*}
K^*_{-1}h_N\not\in L^1(\R),\quad 
(K_{01}+K_{P_1}+K_{P_2})f_N\not\in L^1(\R)\ \ and\ \ (K_{-1}+K_{02})f_{N}\in L^1(\R),
\end{equation*}
Thus, $W^L_{-}\not\in \B(L^{\infty},{\rm BMO})\cup\B(L^1)$.
\end{itemize}
\end{proposition}
\begin{remark}
{\rm Note that in the exceptional case, the asymptotic behavior of the Jost function $f_+(0,x)$ as $x\rightarrow-\infty$ is closely linked to the coefficient $c_2$, which in turn is related to the operators 
$K_{-1},K_{01},K_{02},K_{P_2}$. More precisely, we have (see Remark \ref{rema of S1 and rela betwe c2 and lim})
\begin{equation}\label{rela between c2 and limit}
\lim\limits_{x\rightarrow-\infty}f_{+}(0,x)\not\neq1\Longleftrightarrow c_2\not\neq0,\quad c_2=\frac{1}{2}\int_{\R}yv(y)\omega(y)dy.
 \end{equation}
We remark that the counterexamples above rely on both the non-vanishing of $c_2$ and the compact support of $V$. 
}
\end{remark}
\begin{proof}
(1) From \eqref{kernel of KPj} and the uniform boundedness of $k^{\pm}(x,y)$, we decompose
\begin{align*}
K_{P_1}(x,y)&=-\frac{1}{2}g^{+}(x,y)+O(\left<|x|-|y|\right>^{-2})=-\frac{1}{2}\Big(\mathbf1_E(x,y)+\mathbf1_{E^c}(x,y)\Big)g^{+}(x,y)+O(\left<|x|-|y|\right>^{-2})\\
&=-\frac{1}{2}\mathbf1_E(x,y)\Big(\frac{1}{|x|+|y|}+\frac{1}{|x|-|y|}\Big)+O(\left<|x|-|y|\right>^{-2}),
\end{align*}
where $\mathbf1_E$ is the characteristic function on the set $E=\{(x,y):||x|-|y||\geq2\}$ and $E^c$ corresponds to the complementary set. Denote $$\widetilde{K}_{P_1}(x,y)=-\frac{1}{2}\mathbf1_E(x,y)\big(\frac{1}{|x|+|y|}+\frac{1}{|x|-|y|}\big).$$ It then suffices to prove
$${\widetilde{K}_{P_1}}f_{r}\notin L^1(\R)\ {\rm for\ any}\ r>0\ {\rm and}\  |(\widetilde{K}_{P_1}f_{r})(r+2)|\rightarrow+\infty\ {\rm as}\ \ r\rightarrow +\infty.$$
Indeed, for any $r>0$ and $x\geq r+2$, 
\begin{equation*}
|(\widetilde{K}_{P_1}f_{r})(x)|=\frac{1}{2}\left|\int_{-r}^{r}\chi_{\{||x|-|y||\geq2\}}\Big(\frac{1}{|x|+|y|}+\frac{1}{|x|-|y|}\Big)dy\right|=\left|\int_{0}^{r}\Big(\frac{1}{x+y}+\frac{1}{x-y}\Big)dy\right|=\ln\Big(1+\frac{2r}{x-r}\Big).
\end{equation*}
Consequently,
\begin{equation*}
\|\widetilde{K}_{P_1}f_{r}\|_{L^1}\geq
\int_{r+2}^{+\infty}\ln\Big(1+\frac{2r}{x-r}\Big)dx=+\infty,\qquad
|(\widetilde{K}_{P_1}f_{r})(r+2)|=\ln(1+r)\rightarrow
+\infty\ \  as\ \ r\rightarrow+\infty.
\end{equation*}
\vskip0.3cm
(2) \underline{\bf(i)} To prove $W^L_{-}\not\in\B(L^{\infty},{\rm BMO})$, it suffices to show that $K_{-1}\not\in\B(L^{\infty},{\rm BMO})$, which in turn reduces to
proving $K^*_{-1}\not\in\B(\HH^1,L^1)$ by duality. To derive this, it follows from \eqref{kernels of singular terms} and $\overline{g^{-}(y,x)}=g^+(x,y)$ that
\begin{align*}
K^*_{-1}(x,y)&=\overline{K_{-1}(y,x)}=\overline{\frac{c}{4}a(y)\overline{a(x)}g^{-}(y,x)+O(\left<|y|-|x|\right>^{-2})}\\
&=\frac{c}{4}{a(x)}\overline{a(y)}\mathbf1_E(x,y)\Big(\frac{1}{|x|+|y|}+\frac{1}{|x|-|y|}\Big)+O(\left<|x|-|y|\right>^{-2})
\end{align*}
with $\mathbf1_E$ as in (1) above and 
\begin{equation*}
    a(x)=\int_{\R}\int_{0}^{1}({\rm sgn}(x-\theta y))d\theta yv(y)\omega(y)dy.
\end{equation*}
Thus it is enough to show $\widetilde{K}_{-1}h_N\not\in L^1(\R)$, where 
\begin{equation*}
    \widetilde{K}_{-1}(x,y)=\frac{c}{4}{a(x)}\overline{a(y)}\mathbf1_E(x,y)\Big(\frac{1}{|x|+|y|}+\frac{1}{|x|-|y|}\Big).
\end{equation*}
 Observe that $a(x)=2c_2{\rm sgn}\ x$ when $|x|\geq N$. Hence,
\begin{align*}
\|\widetilde{K}_{-1}h_N\|_{L^1}&\geq\int_{2N+2}^{+\infty}|(\widetilde{K}_{-1}h_N)(x)|dx=2c|c_2|^2\int_{2N+2}^{+\infty}\int_{N}^{2N}\left(\frac{1}{x+y}+\frac{1}{x-y}\right)dydx\\
&=2c|c_2|^2\int_{2N+2}^{+\infty}\ln \frac{(x+2N)(x-N)}{(x-2N)(x+N)}dx.
\end{align*}
Since $\lim\limits_{x\rightarrow-\infty}f_{+}(0,x)\neq1$, i.e., $c_2\neq0$. Therefore, $\|\widetilde{K}_{-1}h_N\|_{L^1}=+\infty$.
\vskip0.3cm
\underline{\bf(ii)} Set 
$$K_1:=K_{01}+K_{P_1}+K_{P_2},\quad K_2:=K_{-1}+K_{02}.$$
From \eqref{kernel of KPj} and \eqref{kernels of singular terms}, we get 
\begin{align*}
K_1(x,y)&=\Big(-\frac{c\overline{c_1}a(x)}{2}+c|c_1|^2-\frac{1}{2}\Big)g^{+}(x,y)+O(\left<|x|-|y|\right>^{-2})\\
&=\Big(-\frac{c\overline{c_1}a(x)}{2}+c|c_1|^2-\frac{1}{2}\Big)\mathbf1_E(x,y)\Big(\frac{1}{|x|+|y|}+\frac{1}{|x|-|y|}\Big)+O(\left<|x|-|y|\right>^{-2})\\
&\xlongequal{\ x\geq N+2}\Big(-c\overline{c_1}c_2+c|c_1|^2-\frac{1}{2}\Big)\mathbf1_E(x,y)\Big(\frac{1}{|x|+|y|}+\frac{1}{|x|-|y|}\Big)+O(\left<|x|-|y|\right>^{-2}).
\end{align*}
Since $c_1+c_2\neq0$ and $c_2\neq0$, using the relation $c=(2(|c_1|^2+|c_2|^2))^{-1}$,  we obtain $$-c\overline{c_1}c_2+c|c_1|^2-\frac{1}{2}=-cc_2\overline{(c_1+c_2)}\neq0.$$ Hence, applying the argument used for $K_{P_1}f_N$ in (1) to $K_1f_N$ gives $K_1f_N\not\in L^1(\R)$.
As for $K_2$, recall from \eqref{kernel of Kminus1} and decompose $k_{-1}(x,y)$ as 
\begin{align*}
k_{-1}(x,y)
&=i({\rm sgn}\ x)({\rm sgn}\ y)\mathbf1_E(x,y)\Big(\frac{1}{|x|+|y|}-\frac{1}{|x|-|y|}\Big)+O(\left<|x|-|y|\right>^{-2})\\
&:=k(x,y)+R(x,y).
\end{align*}
We then have
\begin{equation}\label{kernel of K-1}
K_{-1}(x,y)=\frac{-ic}{4}\int_{\R^2\times[0,1]^2}^{}u_1u_2(vS_1v)(u_1,u_2)k(X_1,Y_2)d\Theta+K_R(x,y):=K(x,y)+K_R(x,y),
\end{equation}
where $\Theta=(u_1,u_2,\theta_1,\theta_2),\ X_1=x-\theta_1u_1,\ Y_2=y-\theta_2u_2$. By the Schur test Lemma \ref{shur test}, we have $K_R\in\B(L^1)$,  so it suffices to show $Kf_N\in L^1(\R)$. Note that for any $x\in\R$ and $-N\leq y\leq N$,
\begin{align*}
|k(x,Y_2)|&=2|({\rm sgn}~x)({\rm sgn}~Y_2)|\chi_{\{||x|-|Y_2||\geq2\}}\frac{|Y_2|}{||x|+|Y_2||\cdot||x|-|Y_2||}\\
&\lesssim|Y_2|\chi_{\{||x|-|Y_2||\geq2\}}||x|-|Y_2||^{-2}\lesssim\left<u_2\right>\left<|x|-|Y_2|\right>^{-2},
\end{align*}
which implies
$$\int_{\R}^{}\int_{-N}^{N}|k(x,Y_2)|dydx\lesssim\left<u_2\right>.$$
Using Minkowski's integral inequality and the translation invariance of the $L^1$ norm, we obtain
\begin{align*}
\|Kf_{N}\|_{L^{1}}&\leq\int_{\R\times[-N,N]}^{}\left|K(x,y)\right|dydx\\
&\lesssim\int_{\R^2\times[0,1]^2}^{}|u_1u_2(vS_1v)(u_1,u_2)|\int_{\R\times[-N,N]}^{}|k(X_1,Y_2)|dydxd\Theta\\
&\lesssim\int_{\R^2}^{}|u_1u_2(vS_1v)(u_1,u_2)|\left<u_2\right>du_1du_2<\infty.
\end{align*}
Thus, $K_{-1}f_{N}\in L^1(\R)$. A similar argument shows $K_{02}f_{N}\in L^1(\R)$, and therefore $K_2f_{N}\in L^{1}(\R)$.
\end{proof}
\appendix
\section{The proof of asymptotic expansion Lemma \ref{asy expa lemma}}\label{proof of asy expansion}
In this appendix, we provide the proof of Lemma \ref{asy expa lemma}, i.e., the expansion of $M^{-1}(\lambda)$. Recall that $$M(\lambda)=U+vR^{+}_{0}(\lambda^{2})v.$$
To begin with, we give asymptotic expansion of $R^{+}_{0}(\lambda^2)$ near $\lambda=0$ in the weighted space $\B(L^2_s(\R),L^2_{-s}(\R))$.
\begin{lemma}
Let $N\geq-1$ be an integer. Then for any $s>\frac{1}{2}+N+2$, the following expansion holds:
\begin{equation}\label{expans of R0} R^{+}_0(\lambda^2)=\sum\limits_{j=-1}^{N}\lambda^jG_j+r_{N+1}(\lambda) \quad in\ \ \B(L^2_s(\R),L^2_{-s}(\R))\quad\lambda\rightarrow0^+,
\end{equation}
where $G_j$ are integral operators with kernels 
\begin{equation*}
    G_{-1}(x,y)=\frac{i}{2},\quad G_0(x,y)=-\frac{1}{2}|x-y|, \quad G_1(x,y)=\frac{- i}{4}|x-y|^2,\quad
    G_j(x,y)=-\frac{i^{j}}{2(j+1)!}|x-y|^{j+1},
\end{equation*}
and $\|r_{N+1}(\lambda)\|_{L^2_s(\R)\rightarrow L^2_{-s}(\R)}=O(\lambda^{N+1})$. Moreover, in the same sense, \eqref{expans of R0} can be differentiated $N+2$ times.
\end{lemma}
\begin{proof}
For any integer $N\geq-1$, using \eqref{kernel of free resol} and Taylor expansion with remainders, we obtain
\begin{equation*}
R^+_0(\lambda^2,x,y)=-\sum\limits_{j=-1}^{N}\frac{i^{j}}{2(j+1)!}|x-y|^{j+1}\lambda^j+r_{N+1}(\lambda,x,y),\quad \lambda\rightarrow 0^+.
\end{equation*}
One readily checks that 
\begin{equation*}
r_{N+1}(\lambda,x,y)=\lambda^{N+1}b_N(\lambda)|x-y|^{N+2},\quad b_N(\lambda)=O(1),\quad \lambda\rightarrow 0^+.
\end{equation*}
Observe that if $s>\frac{1}{2}+N+2$, using the estimate $|x-y|\lesssim\left<x\right> \left<y\right>$, we have
\begin{equation*}
\int_{\R^2}\left<x\right>^{-2s}|x-y|^{2N+4} \left<y\right>^{-2s}dxdy\lesssim \|\left<x\right>^{2N+4-2s}\|^2_{L^1}<\infty.
\end{equation*}
Consequently, the operators $G_j$ and $r_{N+1}(\lambda)$ are Hilbert-Schmidt operators, and the expansion \eqref{expans of R0} follows. By a similar argument, one can verify  that \eqref{expans of R0} can be differentiated $N+2$ times.
\end{proof}
Moreover, we recall that
$$P=\alpha\left<\cdot,v\right>v,\quad Q=I-P,\quad \alpha=\|V\|^{-1}_{L^1}$$
and 
$$S_1L^2(\R)=\{f\in L^2(\R):\left<f,v\right>=0,\ QT_0f=0\},\quad T_0=U+vG_0v.$$
\begin{remark}\label{rema of S1 and rela betwe c2 and lim}
{\rm (1) It is known from \cite[Lemma 5.4]{JN01} that dim($S_1L^2(\R))\leq1$. Moreover, $S_1L^2(\R)\neq\{0\}$ if and only if there exists a unique (up to a constant factor) non-zero $\phi\in L^{\infty}(\R)$ such that $H\phi=0$ in the distributional sense, and 
\begin{equation*}
\lim\limits_{x\rightarrow+\infty}\phi(x)=c_1-c_2\neq0,\quad \lim\limits_{x\rightarrow-\infty}\phi(x)=c_1+c_2\neq0,\quad c_1=\alpha\left<T_0\omega,v\right>,\quad c_2=\frac{1}{2}\int_{\R}yv(y)\omega(y)dy,
\end{equation*}
where $\omega$ is an orthonormal basis of $S_1L^2(\R)$. This yields the characterization:
$$V\  {\rm is\  of\  exceptional\  type}\   \Longleftrightarrow S_1L^2(\R)\neq\{0\}.$$ 
(2) Moreover, we note that if $V$ is of exceptional type, then the Jost functions are given by
 $$f_{+}(0,x)=\frac{\phi(x)}{c_1-c_2},\quad f_-(0,x)=\frac{\phi(x)}{c_1+c_2},$$
and consequently,
 \begin{equation*}
\lim\limits_{x\rightarrow-\infty}f_{+}(0,x)=1\Longleftrightarrow c_2=0.
 \end{equation*}
}   
\end{remark}
In order to derive the expansion, the following lemma will be used frequently.
\begin{lemma}\label{lemm of inverse}
{\rm (\cite[Lemma 2.1]{JN01})} Let $\mscH$ be a complex Hilbert space. Let $A$ be a closed operator and $S$ a projection. Suppose $A+S$ has a bounded inverse. Then $A$ has a bounded inverse if and only if
$$\mcaA\equiv S-S(A+S)^{-1}S$$
has a bounded inverse in $S\mscH$, and in this case
$$A^{-1}=(A+S)^{-1}+(A+S)^{-1}S \mcaA^{-1}S(A+S)^{-1}.$$
\end{lemma}
Now we give the proof of Lemma \ref{asy expa lemma}.
\begin{proof}[Proof of Lemma {\rm\ref{asy expa lemma}}]
{\bf(1) Assume that $V$ is of exceptional type and $\beta>11$.} By Remark \ref{rema of S1 and rela betwe c2 and lim}, we have $S_1\neq0$ and $S_1=\left<\cdot,\omega\right>\omega$ with $\|\omega\|_{L^2}=1$. 
\vskip0.2cm
\underline{\bf Step 1:} Taking $N=3$ in \eqref{expans of R0} gives
\begin{equation*}
M(\lambda)=U+vR^{+}_0(\lambda^2)v=a\lambda^{-1}\Big(P+\frac{1}{a}T_0\lambda+\frac{1}{a}\sum\limits_{j=2}^{4}vG_{j-1}v\lambda^j+\widetilde{r}_5(\lambda)\Big):=a\lambda^{-1}\widetilde{M}(\lambda),\quad \lambda\rightarrow0^+,
\end{equation*}
where $a=\frac{i}{2\alpha}$ and $\|\widetilde{r}_5(\lambda)\|_{\B(L^2)}=O(\lambda^5)$. By the von Neumann series expansion,
\begin{equation}\label{inverse of M+Q}
    (\widetilde{M}(\lambda)+Q)^{-1}=I+\sum\limits_{j=1}^{4}B_j\lambda^j
+O(\lambda^5),\quad \lambda
\rightarrow0^+
\end{equation}
with
\begin{equation*}
    B_1=-\frac{1}{a}T_0,\quad B_2=\frac{1}{a^2}T_0^2-\frac{1}{a}vG_1v.
\end{equation*}
Define $$M_1(\lambda):=Q-Q(\widetilde{M}(\lambda)+Q)^{-1}Q=\frac{1}{a}\lambda \Big(QT_0Q-a\sum\limits_{j=1}^3QB_{j+1}Q\lambda^j+QO(\lambda^4)Q\Big):=\frac{1}{a}\lambda \widetilde{M}_1(\lambda).$$ Applying Lemma \ref{lemm of inverse} to $\widetilde{M}(\lambda)$, we obtain
\begin{equation}\label{inverse of M}
    M^{-1}(\lambda)=a^{-1}\lambda\widetilde{M}^{-1}(\lambda)=a^{-1}\lambda(\widetilde{M}(\lambda)+Q)^{-1}\Big[I+a\lambda^{-1}Q\widetilde{M}^{-1}_1(\lambda)Q(\widetilde{M}(\lambda)+Q)^{-1}\Big].
\end{equation}
Since ${\rm Ker}QT_0Q|_{QL^2(\R)}=S_1L^2(\R)\neq\{0\}$, we proceed with the {\bf Step 2} to compute the inverse of $\widetilde{M}^{-1}_1(\lambda)$.
\vskip0.2cm
\underline{\bf Step 2:} By definition of $S_1$,  $QT_0Q+S_1$ is invertible on $QL^2(\R)$. Let $D_0:=(QT_0Q+S_1)^{-1}$. Then 
\begin{equation}\label{inverse of M1+S1}
(\widetilde{M}_1(\lambda)+S_1)^{-1}=D_0+\sum\limits_{j=1}^3C_j\lambda^j+D_0QO(\lambda^4)QD_0,\quad C_1=aD_0QB_2QD_0.
\end{equation}
Set $$M_2(\lambda):=S_1-S_1(\widetilde{M}_1(\lambda)+S_1)^{-1}S_1=\frac{i}{c}\lambda\Big(S_1+ic\sum\limits_{j=2}^{3}S_1C_jS_1\lambda^{j-1}+S_1O(\lambda^3)S_1\Big),$$
where we used the fact (cf. \cite[Lemma 5.4]{JN01}) that $-S_1C_1S_1=\frac{i}{c}S_1$, with 
\begin{equation}\label{c,c1,c2}
c=\frac{1}{2(|c_1|^2+|c_2|^2)},\quad c_1=\alpha \left<T_0\omega,v\right>,\quad c_2=\frac{1}{2}\int_{\R}yv(y)\omega(y)dy.
\end{equation}
Thus $M_2(\lambda)$ is invertible on $S_1L^2(\R^2)$ and 
\begin{align*}
    M^{-1}_2(\lambda)=-ic\lambda^{-1}\Big(S_1-icS_1C_2S_1\lambda-\big(icS_1C_3S_1+c^2(S_1C_2S_1)^2\big)\lambda^2 +S_1O(\lambda^3)S_1\Big).
\end{align*}
Applying Lemma \ref{lemm of inverse} to $\widetilde{M}_1(\lambda)$ and using \eqref{inverse of M1+S1}, we get 
\begin{equation*}
\widetilde{M}^{-1}_1(\lambda)=(\widetilde{M}_1(\lambda)+S_1)^{-1}\Big[I+S_1M^{-1}_2(\lambda)S_1(\widetilde{M}_1(\lambda)+S_1)^{-1}\Big]=-icS_1\lambda^{-1}+E_0+E_1\lambda+D_0O(\lambda^2),
\end{equation*}
with $E_0=D_0-ic(S_1C_1+C_1S_1)-c^2S_1C_2S_1$ and
\begin{equation*}
    E_1=C_1-ic\big(S_1C_2+C_2S_1-(cS_1C_2S_1)^2+C_1S_1C_1\big)-c^2(S_1C_2S_1C_1+C_1S_1C_2S_1+S_1C_3S_3).
\end{equation*}
Substituting this and \eqref{inverse of M+Q} into \eqref{inverse of M}, we obtain \eqref{asy expa res}.
\vskip0.3cm
{\bf (2) Assume that $V$ is of generic type and $\beta>7$.} In this case,  $S_1L^2(\R)=\{0\}$, i.e., $QT_0Q$ is invertible on $QL^2(\R)$. Taking $N=1$ in \eqref{expans of R0} and repeating the {\bf Step 1} from (1), we obtain 
\begin{equation*}
    M^{-1}(\lambda)=a^{-1}\lambda(\widetilde{M}(\lambda)+Q)^{-1}\Big[I+a\lambda^{-1}Q\widetilde{M}^{-1}_1(\lambda)Q(\widetilde{M}(\lambda)+Q)^{-1}\Big]
\end{equation*}
with  
\begin{equation*}
    (\widetilde{M}(\lambda)+Q)^{-1}=I+B_1\lambda+B_2\lambda^2
+O(\lambda^3),\quad \lambda
\rightarrow0^+
\end{equation*}
and 
\begin{equation*}
\widetilde{M}^{-1}_1(\lambda)=D_0+aD_0QB_2D_0\lambda+D_0QO(\lambda^2)QD_0,\quad \lambda
\rightarrow0^+. 
\end{equation*}
A direct calculation then yields the desired expansion \eqref{asy expa reg}.
\end{proof}

\section{The proof of $L^p$ boundedness of the high energy part $W^H_{-}$}\label{subsec of high energy part}
For completeness and as a different method from Weder\cite{Wed99}, in this appendix, we present the proof of the $L^p$ boundedness of the high energy part $W^H_{-}$ in our setting.
\begin{theorem}\label{bounds of high energy}
 Let $H=-\Delta+V$ with $|V(x)|\lesssim \left<x\right>^{-\beta}$ for $\beta>6$. Then $W^{H}_{-}\in\B(L^p)$ for all $1\leq p\leq\infty$.
\end{theorem}
Recalling the formula \eqref{high energy part} for $W^H_{-}$ and using the resolvent identity
\begin{equation*} R^+_V(\lambda^2)=R^+_0(\lambda^2)-R^+_0(\lambda^2)VR^+_V(\lambda^2),
\end{equation*}  
we decompose $W^{H}_{-}$ as
$$W^{H}_{-}=\int_{0}^{+\infty}\lambda(1-\chi(\lambda))R^{+}_{V}(\lambda^{2})V(R^+_{0}-R^-_{0})({\lambda}^{2})d\lambda=W^H_{1}-W^H_2,$$
where $\chi(\lambda)\in C^{\infty}_{0}(\R)$ satisfies $\chi(\lambda)\equiv1 $ on $(-\lambda_{0},\lambda_{0})$ and  $\chi(\lambda)\equiv0 $ on $(-\infty,-2\lambda_{0})\cup(2\lambda_{0},+\infty)$, and
\begin{align*}
W^H_{1}&=\int_{0}^{+\infty}\lambda(1-\chi(\lambda))R^{+}_{0}(\lambda^{2})V(R^+_{0}-R^-_{0})({\lambda}^{2})d\lambda,\\
W^H_{2}&=\int_{0}^{+\infty}\lambda(1-\chi(\lambda))R^{+}_{0}(\lambda^{2})VR^+_V(\lambda^2)V(R^+_{0}-R^-_{0})({\lambda}^{2})d\lambda.
\end{align*}
We first prove the $L^p$ boundedness of $W^H_{1}$. For this purpose, the following lemma is crucial.
\begin{lemma}\label{lemma of estim of u(x)}
For any $x\in\R\setminus\{0\}$, denote 
\begin{equation}\label{integral u}
u(x):=\int_{0}^{+\infty}e^{i\lambda x}\lambda^{-1}(1-\chi(\lambda)) d\lambda.
\end{equation}
Then 
\begin{equation}\label{estim of u(x)}
|u(x)|\lesssim\min\{|\ln|x||,|x|^{-2}\}\in L^1(\R).
\end{equation}
\end{lemma} 
The integral \eqref{integral u} should be understood in the distributional sense, i.e., for any Schwartz function $\psi$, 
\begin{equation*}
u(\psi):=\int_{\R}u(x)\psi(x)dx=\int_{0}^{+\infty}\lambda^{-1}(1-\chi(\lambda))\hat{\psi}(-\lambda)d\lambda, \quad \hat{\psi}(\lambda):=\int_{\R}e^{-i\lambda x}\psi(x)dx.
\end{equation*}
Thus $u\in \mcaS'(\R)$ (the space of tempered distributions) because
\begin{equation*}
|u(\psi)|\lesssim\big\|x\hat{\psi}(x)\big\|_{L^{\infty}}=\big\|\hat{\psi'}\big\|_{L^{\infty}}\leq \|\psi'\|_{L^1}\lesssim \sum\limits_{0\leq\alpha\leq3}\|x^{\alpha}\psi'(x)\|_{L^{\infty}}.
\end{equation*}
The proof of this lemma is postponed to the end of this appendix.
\begin{proposition}\label{bound of W1}
Let $H=-\Delta+V$ with $|V(x)|\lesssim \left<x\right>^{-\beta}$ for $\beta>1$. Then $W^{H}_1\in\B(L^p)$ for all $1\leq p\leq\infty$.
\end{proposition}
\begin{proof}
By \eqref{kernel of free resol}, the kernel of $W^H_1$ is given by 
\begin{equation*}
    W^{H}_{1}(x,y)=-\frac{1}{4}\sum\limits_{\pm}\int_{0}^{+\infty}\lambda^{-1}(1-\chi(\lambda)) \int_{\R}V(\omega)e^{i\lambda(|x-\omega|\pm|y-\omega|)}d\omega d\lambda:=-\frac{1}{4}\sum\limits_{\pm}W^{\pm}_1(x,y).
\end{equation*}
Denote 
\begin{equation*}
    f(x)=\min\{|\ln x|,x^{-2}\},\quad x>0.
\end{equation*}
Note that for any $\omega,y\in\R$, using the invariance of the $L^1$-norm and and the evenness of $f(||\cdot|\pm|y||)$, we get
\begin{equation*}
\big\|f(||\cdot-\omega|\pm|y||)\big\|_{L^1}=\left\{\begin{aligned}
&2\int_{|y|}^{+\infty}f(x)dx\leq 2\int_{0}^{+\infty}f(x)dx<\infty,\quad ``+"\\
&2\Big(\int_{0}^{|y|}f(|y|-x)dx+\int_{|y|}^{+\infty}f(x-|y|)dx\Big)\leq 4\int_{0}^{+\infty}f(x)dx<\infty, \quad ``-".
\end{aligned}\right.
\end{equation*}
This, together with Lemma \ref{lemma of estim of u(x)} gives  
\begin{equation*}
\int_{\R}|W^{\pm}_1(x,y)|dx\leq \int_{\R^2}\big|V(\omega)u(|x-\omega|\pm|y-\omega|)\big|d\omega dx \lesssim\int_{\R}|V(\omega)|\cdot\big\|f(||\cdot-\omega|\pm|y-\omega||)\big\|_{L^1}d\omega\lesssim\|V\|_{L^1}<\infty.
\end{equation*}
By symmetry, the same estimate holds for $\|W^{\pm}_1(x,\cdot)\|_{L^1}$. Lemma \ref{shur test} then yields $W^H_1\in\B(L^p)$ for all $1\leq p\leq\infty$. 
\end{proof}
To treat $W^H_2$, we first give the following estimates of $R^+_V(\lambda)$ for large $\lambda$ based on \cite[Lemma 3.9]{FSWY20}. 
\begin{lemma}\label{lemm of RV at high energy}
For $k=0,1,2,3,\cdots$, let $H=-\Delta+V$ with $|V(x)|\lesssim \left<x\right>^{-\beta}$ for $\beta>2(k+1)$. Then $(\partial^k_{\lambda}R^{+}_V)(\lambda)\in\B(L^2_s(\R),L^2_{-s}(\R))$ is continuous for any $s>k+\frac{1}{2}$. Moreover, the following estimate holds:    
 \begin{equation*}
\big\|(\partial^k_{\lambda}R^{+}_V)(\lambda)\big\|_{\B(L^2_s(\R),L^2_{-s}(\R))} \lesssim \lambda^{-\frac{(1+k)}{2}},\quad s>k+\frac{1}{2},\quad \lambda\geq\lambda^2_0. 
 \end{equation*}
\end{lemma} 
\begin{proposition}\label{bound of W2}
Let $H=-\Delta+V$ with $|V(x)|\lesssim \left<x\right>^{-\beta}$ for $\beta>6$. Then $W^{H}_2\in\B(L^p)$ for all $1\leq p\leq\infty$.
\end{proposition}
\begin{proof}
It follows from \eqref{kernel of free resol} that
\begin{equation*}
    W^{H}_{2}(x,y)=-\frac{1}{4}\sum\limits_{\pm}W^{\pm}_2(x,y),
\end{equation*}
where $\widetilde{\chi}(\lambda)=\lambda^{-1}(1-\chi(\lambda))$ and 
\begin{align*}
W^{\pm}_2(x,y)&=\int_{0}^{+\infty}e^{i\lambda(|x|\pm|y|)}\widetilde{\chi}(\lambda)k^{\pm}_{H}(\lambda,x,y)d\lambda,\\
k^{\pm}_{H}(\lambda,x,y)&=\big<R^+_V(\lambda^2)F^{\pm}_{y,\lambda}, F^-_{x,\lambda}\big>,
\quad F^{\pm}_{y,\lambda}(\omega)=e^{\pm i\lambda(|y-\omega|-|y|)}V(\omega).
\end{align*}
Note that using the triangle inequality and the decay of $V$, we obtain, for any $s<\beta-\frac{1}{2}-k$,
\begin{equation*}
\sup\limits_{\lambda,y}\big\|\partial^k_{\lambda}F^{\pm}_{y,\lambda}\big\|_{L^2_s(\R)}\lesssim1.
\end{equation*}
Moreover, Lemma \ref{lemm of RV at high energy} implies that for $j=0,1,2$ and $s>j+\frac{1}{2}$,
\begin{equation*}
\|\partial^{j}_{\lambda}\big(R^+_V(\lambda^2)\big)\|_{\B(L^2_s(\R),L^2_{-s}(\R))} \lesssim \lambda^{-1},\quad \lambda\geq\lambda_0.
\end{equation*}
Combining these two estimates and the formula
\begin{align*}
\big(\partial^{\ell}_{\lambda}k^{\pm}_{H}\big)(\lambda,x,y)=\sum\limits_{k=0}^{\ell}C^k_{\ell}\sum\limits_{j=0}^{k}C^j_k\big<\partial^j_{\lambda}\big(R^+_V(\lambda^2)\big)\cdot\big(\partial^{k-j}_{\lambda}F^{\pm}_{y,\lambda}\big),\partial^{\ell-k}_{\lambda}F^{-}_{x,\lambda}\big>, 
\end{align*}
we derive, for $\beta>6$ and $\ell=0,1,2$,
\begin{equation}\label{esti of kpmH}
\sup\limits_{x,y}\big|\big(\partial^{\ell}_{\lambda}k^{\pm}_{H}\big)(\lambda,x,y)\big|\lesssim\lambda^{-1},\quad \lambda\geq\lambda_0.
\end{equation}
It is clear that $|W^{\pm}_2(x,y)|\lesssim1\lesssim\left<|x|-|y|\right>^{-2}$ if $||x|\pm|y||\leq1$. For $||x|\pm|y||\geq1$, denote $K^{\pm}_{H}(\lambda,x,y):=\widetilde{\chi}(\lambda)k^{\pm}_{H}(\lambda,x,y).$ Note that for $j=0,1,2$, 
\begin{equation}\label{estim of chituta}
\widetilde{\chi}^{(j)}(\lambda)\big|_{\lambda=\lambda_0}=0,\quad \big|\widetilde{\chi}^{(j)}(\lambda)\big|\lesssim\lambda^{-(j+1)},\quad \lambda\geq\lambda_0. 
\end{equation}
Integrating by parts twice to $W^{\pm}_2(x,y)$ and using the estimates \eqref{estim of chituta} and \eqref{esti of kpmH} yields
\begin{align*}
|W^{\pm}_2(x,y)|=\frac{1}{(|x|\pm|y|)^2}\Big|\int_{\lambda_0}^{+\infty}e^{i\lambda(|x|\pm|y|)}(\partial^2_{\lambda}K^{\pm}_{H})(\lambda,x,y)d\lambda\Big|\lesssim\frac{1}{(|x|\pm|y|)^2}\int_{\lambda_0}^{+\infty}\lambda^{-2}d\lambda\lesssim\left<|x|-|y|\right>^{-2}.   
\end{align*}
Therefore, by Lemma \ref{shur test}, we conclude $W^{H}_2\in\B(L^p)$ for all $1\leq p\leq \infty$.
\end{proof}
Combining Propositions \ref{bound of W1} and \ref{bound of W2}, we complete the proof of Theorem \ref{bounds of high energy}. Finally, we come to prove Lemma \ref{lemma of estim of u(x)}.
\begin{proof}[Proof of Lemma {\rm\ref{lemma of estim of u(x)}}]

To establish \eqref{estim of u(x)}, for any $\varepsilon>0$, we introduce
\begin{equation*}
    u_{\varepsilon}(x):=\int_{0}^{+\infty}e^{i\lambda x}\widetilde{\chi}(\lambda)\varphi(\varepsilon \lambda) d\lambda,\quad \widetilde{\chi}(\lambda)=\lambda^{-1}(1-\chi(\lambda)),
\end{equation*}
where $\varphi\in C^{\infty}(\R)$ satisfies $\varphi(x)=1$ for $|x|\leq\frac{1}{2}$ and $\varphi(x)=0$ for $|x|\geq1$.
Clearly, $u_{\varepsilon}\in \mcaS'(\R)$ for any $\varepsilon>0$. Moreover, we note that for any $\psi\in\mcaS(\R)$, the dominated convergence theorem yields
\begin{equation*}
    u_{\varepsilon}(\psi)=\int_{0}^{+\infty}\widetilde{\chi}(\lambda)\varphi(\varepsilon \lambda)\hat{\psi}(-\lambda)d\lambda\rightarrow\int_{0}^{+\infty}\widetilde{\chi}(\lambda)\hat{\psi}(-\lambda)d\lambda=u(\psi),\quad \varepsilon\rightarrow0.
\end{equation*}
Next we claim the uniform estimate
\begin{equation}\label{point estim of ue}
\sup\limits_{\varepsilon>0}|u_{\varepsilon}(x)|\lesssim \min\{|\ln|x||,|x|^{-2}\}:=f(x).
\end{equation}
Once this is proved, since $f\in L^1(\R)$, for any $\psi\in C^{\infty}_0(\R)$, we obtain
\begin{equation*}
|u(\psi)|=\lim\limits_{\varepsilon\rightarrow0}|u_{\varepsilon}(\psi)|\leq\limsup\limits_{\varepsilon\rightarrow0}\int_{\R}|u_{\varepsilon}(x)\psi(x)|dx\lesssim\int_{\R}f(x)|\psi(x)|dx\lesssim \sup\limits_{x\in\R}|\psi(x)|
\end{equation*}
which implies the desired \eqref{estim of u(x)}.
To obtain \eqref{point estim of ue}, we distinguish two cases. If $0<|x|<1$, since supp~$(1-\chi(\lambda))\subseteq[\lambda_0,+\infty)$ and supp~$\varphi(\lambda)\subseteq\{|\lambda|\leq1\}$, we have 
\begin{equation*}
u_{\varepsilon}(x)=\Big(\int_{\lambda_0}^{|x|^{-1}}+\int_{|x|^{-1}}^{\varepsilon^{-1}}\Big)e^{i\lambda x}\lambda^{-1}(1-\chi(\lambda))\varphi(\varepsilon \lambda) d\lambda:=u_{\varepsilon,1}(x)+u_{\varepsilon,2}(x).
\end{equation*}
For $u_{\varepsilon,1}(x)$, 
\begin{equation*}
 |u_{\varepsilon,1}(x)|\lesssim \int_{\lambda_0}^{|x|^{-1}}\lambda^{-1}d\lambda \lesssim |\ln|x||. 
\end{equation*}
For $u_{\varepsilon,2}(x)$, integrating by parts once yields
\begin{equation*}
 u_{\varepsilon,2}(x)=i\frac{e^{i\frac{x}{|x|}}}{x}\widetilde{\chi}(|x|^{-1})\varphi(\varepsilon|x|^{-1})-\frac{1}{ix}\int_{|x|^{-1}}^{\varepsilon^{-1}}e^{i\lambda x}(\widetilde{\chi}(\lambda)\varphi(\varepsilon\lambda))' d\lambda.
\end{equation*}
Notice that by the estimate \eqref{estim of chituta} and $\varphi\in C^{\infty}_0(\R)$, we get
$$\big|(\widetilde{\chi}(\lambda)\varphi(\varepsilon\lambda))'\big|=\big|\widetilde{\chi}'(\lambda)\varphi(\varepsilon\lambda)+\lambda^{-1}\widetilde{\chi}(\lambda)\varepsilon\lambda\varphi'(\varepsilon\lambda)\big|\lesssim\lambda^{-2},\quad \lambda\geq\lambda_0.$$
Therefore,
$$ |u_{\varepsilon,2}(x)|\lesssim 1+\frac{1}{|x|}\int_{|x|^{-1}}^{+\infty}\lambda^{-2}dx\lesssim1,\quad {\rm uniformly\ in \ }\varepsilon>0.$$
If $|x|\geq1$, applying the integration by parts twice to $u_{\varepsilon}(x)$ and using  \eqref{estim of chituta} and $\varphi\in C^{\infty}_0(\R)$, we obtain 
\begin{equation*}
|u_{\varepsilon}(x)|
=\Big|-x^{-2}\int_{\lambda_0}^{+\infty}e^{i\lambda x}(\widetilde{\chi}(\lambda)\varphi(\varepsilon\lambda))'' d\lambda\Big|\lesssim|x|^{-2}\int_{\lambda_0}^{+\infty}\lambda^{-3}d\lambda\lesssim|x|^{-2}.
\end{equation*}
To sum up, we derive \eqref{point estim of ue}. This completes the proof.
\end{proof}

\vskip0.5cm
{\bf Acknowledgements:} The both authors are partially supported by NSFC (No. 12171182 and 12531005)

\normalem

\end{document}